\AddToHook{package/hyperref/before}{\RequirePackage{float}}
\documentclass[12pt]{iopjournal}
\usepackage[T1]{fontenc}
\usepackage[utf8]{inputenc}
\usepackage{newtxtext}
\usepackage{amsmath,amssymb}
\usepackage{amsthm}
\newtheorem{theorem}{Theorem}
\newtheorem{lemma}{Lemma}
\newtheorem{proposition}{Proposition}
\newtheorem{corollary}{Corollary}
\usepackage{newtxmath}
\usepackage{bm}
\newcommand{\defeq}{\mathrel{:=}}
\usepackage{microtype,etoolbox}
\makeatletter
\IfPackageAtLeastTF{microtype}{2025/02/12}{}{%
  \patchcmd{\MT@setup@expansion}
    {\MT@ifdefined@n@TF{showhyphens }}
    {\MT@exp@cs\patchcmd{showhyphens }
       {\color@begingroup}{\color@begingroup\pdfadjustspacing\z@}
       {\MT@info@nl{Patching command `\string\showhyphens'}}
       {\PackageError{iop_preamble}{Cannot patch \string\showhyphens}
          {Check the microtype compatibility block in iop_preamble.tex.}}%
     \let\MT@temp\@gobbletwo\@gobbletwo}
    {}{\PackageError{iop_preamble}{Cannot apply the microtype compatibility fix}
         {Check the installed microtype version before compiling.}}%
}
\makeatother
\usepackage{booktabs,tabularx,longtable,array,float,makecell}
\usepackage{placeins}
\usepackage[all]{hypcap}
\usepackage[authoryear,round]{natbib}
\usepackage{xurl}
\setcitestyle{authoryear,round,semicolon,aysep={},yysep={,}}
\newcolumntype{Y}{>{\raggedright\arraybackslash}X}
\renewcommand{\articletype}[1]{}
\fancyhfoffset[L]{0pt}
\hypersetup{colorlinks=true,allcolors=blue,pdfauthor={}}

\newcommand{\PaperAuthorBlock}{%
\author{Bing-Zhou Chen$^1$, Amelia Deakin$^2$, Andy Adler$^3$, Paul D. Ledger$^4$, William R. B. Lionheart$^{5,*}$}\par
\affil{$^1$College of Information Engineering, Tarim University, Alar, Xinjiang 843300, China}\par
\affil{$^2$The University of Manchester, Manchester, United Kingdom}\par
\affil{$^3$Systems and Computer Engineering, Carleton University, Ottawa, Ontario K1S 5B6, Canada}\par
\affil{$^4$School of Computing and Mathematical Sciences, University of Leicester, Leicester LE1 7RH, United Kingdom}\par
\affil{$^5$Department of Mathematics, The University of Manchester, Manchester M13 9PL, United Kingdom}\par
\affil{$^*$Corresponding author.}\par
\email{\href{mailto:bill.lionheart@manchester.ac.uk}{bill.lionheart@manchester.ac.uk}}\par}
\newcommand{\PaperAcknowledgments}{William R. B. Lionheart, Paul D. Ledger and Andy Adler would like to thank the Isaac Newton Institute for Mathematical Sciences, Cambridge, for support and hospitality during the programme RNT, where work on this paper was undertaken.}
\newcommand{\PaperFunding}{This work was supported by the Engineering and Physical Sciences Research Council (EPSRC) under grants EP/Z000580/1 and EP/V007742/1.}
\newcommand{\PaperAuthorContributions}{}
\newcommand{\PaperCompetingInterests}{The authors declare no competing interests.}
\newcommand{\PaperAIStatement}{During the preparation of this manuscript, the authors used ChatGPT (OpenAI; GPT-6.1 Sol) to assist with translation into English. The AI-assisted translation was manually checked for linguistic accuracy and consistency with the intended scientific meaning, and revised where necessary. The authors accept full responsibility for the manuscript's content.}
\newcommand{\PaperDataAvailability}{}

\graphicspath{{Source/}}

\usepackage{needspace}
\hypersetup{pdftitle={Implementing the EIT complete electrode model through Robin boundary conditions},pdfsubject={Preprint}}
\begin{document}
\articletype{Paper}
\title{Implementing the EIT complete electrode model through Robin boundary conditions}
\PaperAuthorBlock

\keywords{electrical impedance tomography, complete electrode model, Robin boundary conditions, transadmittance matrix, finite element method, numerical verification}\par

\begin{abstract}
\textbf{Objective.} Medical electrical impedance tomography systems prescribe electrode currents. Accurate voltages, even on passive electrodes, require the complete electrode model (CEM), the most accurate model known for skin electrodes. Its current constraints specify total electrode currents, whereas standard finite element boundary interfaces impose local conditions. We describe how to implement the CEM through Robin boundary conditions, which are available in standard FEM packages.

\textbf{Approach.} Prescribed electrode potentials give Robin problems whose unit-potential responses define an electrode transadmittance matrix. We prove equivalence to the conventional augmented CEM and formalise the discrete energy, equivalence and reference arguments in Lean 4. Numerical comparisons use two-dimensional P1 finite elements in DOLFINx, FreeFEM, GetFEM, MFEM and NGSolve, with the same nodal reference potential as EIDORS.

\textbf{Main results.} For identical assembled blocks, the two procedures have the same solution set whenever the Robin matrix is invertible. Across the five packages in dimensionless test cases, maximum absolute electrode-potential differences from EIDORS were below $5\times10^{-13}$ (2692 nodes), and errors against an exact rational-arithmetic reference were below $1.1\times10^{-13}$ (2535 nodes).

\textbf{Significance.} Standard Robin solves and consistent electrode-current integration provide a portable route to the current-driven CEM, with current constraints enforced through a small electrode system.
\end{abstract}

\section{Introduction}
Electrical impedance tomography (EIT) estimates the conductivity inside an object, or its changes, from currents applied and voltages measured through surface electrodes. Medical EIT systems generally prescribe the applied currents for safety reasons, since patient currents must stay within the limits set by safety standards \citep{iec60601}. The forward model therefore takes the electrode currents as input and must determine the electrode voltages.

The complete electrode model (CEM) is the most accurate model currently known for conductive electrodes applied to the skin: it accounts for electrode size, the shunting effect of the highly conductive electrodes and the electrode--skin contact impedance, with established existence and uniqueness results \citep{cheng1989electrode,somersalo1992existence}. Passive electrodes also require the CEM for accurate voltages. Although their total current is zero, their local normal current density need not vanish: current can enter through one part of an electrode and leave through another (figure~\ref{fig:graphical_abstract_zh}(a)).

Besides dedicated finite element implementations such as EIDORS \citep{vauhkonen1999three,polydorides2002matlab,adler2006eidors}, CEM solvers include boundary element methods for piecewise-constant media, whose accuracy and stability have been compared with those of the method of fundamental solutions \citep{aykroyd2007boundary,dyhoum2014solving}, and a boundary integral formulation with adaptive Nystr\"om discretisation, tested in forward and inverse problems with experimental EIT data \citep{tyni2024boundary}. 
General-purpose FEM packages can provide access to scalable solvers and a broad range of modelling capabilities.
Therefore, we were motivated to develop a way to use general-purpose finite-element software to solve the CEM.

Our approach uses two steps: a solution via Robin boundary conditions (available in many general-purpose FEM packages) and then a transformation step via a transadmittance matrix. An early version of our results was shown at EIT2026 \citep{deakin2026robin}. Prescribing the electrode potentials reduces the CEM to standard Robin problems; their unit-potential responses define an electrode transadmittance matrix that maps electrode potentials to currents, and solving this small electrode system for prescribed currents determines the potentials (figure~\ref{fig:graphical_abstract_zh}). All auxiliary problems share one Robin matrix and can use the package's existing solvers and preconditioners. The finite element subproblem requires no CEM-specific functionality; consistent boundary integration supplies the electrode currents.

This study proves the discrete equivalence of this construction and the conventional augmented CEM before any potential reference is chosen, and formalises the discrete energy argument, the equivalence theorem and the reference results in Lean~4 \citep{demoura2021lean4}. It then relates alternative potential references and adopts the EIDORS nodal reference for controlled numerical comparisons. Implementations through the native Robin interfaces of five platforms demonstrate portability, and shared assembled matrices and an independently assembled exact rational reference separate solution-procedure effects from assembly-induced errors.

Reproducibility is supported by independent finite element implementations, exact rational-arithmetic reference solutions and Lean~4 formal verification of the discrete mathematics, with the corresponding code, fixed inputs and proof sources supplied in the reproduction package.

Section~\ref{sec:theory} presents the two solution procedures, their discrete equivalence and the treatment of potential references. Section~\ref{sec:numerical} describes the comparison design and numerical error measures. Sections~\ref{sec:results}--\ref{sec:conclusions} report the results, discuss their interpretation and summarise the practical implications.

\section{Theory and methods}\label{sec:theory}

An overview of our method is shown in figure~\ref{fig:graphical_abstract_zh}. The CEM is traditionally implemented by augmenting the finite element system with one unknown potential per electrode (figure~\ref{fig:graphical_abstract_zh}(b), left). Our approach instead solves standard Robin problems with prescribed electrode potentials (figure~\ref{fig:graphical_abstract_zh}(b), right) and combines their responses through a small electrode transadmittance system to recover the CEM solution (figure~\ref{fig:graphical_abstract_zh}(c)). Together, these steps are equivalent to the augmented formulation for the same discrete CEM blocks and an invertible Robin matrix, as established in section~\ref{sec:equivalence}.

\begin{figure}[p]
\centering\includegraphics[width=\linewidth]{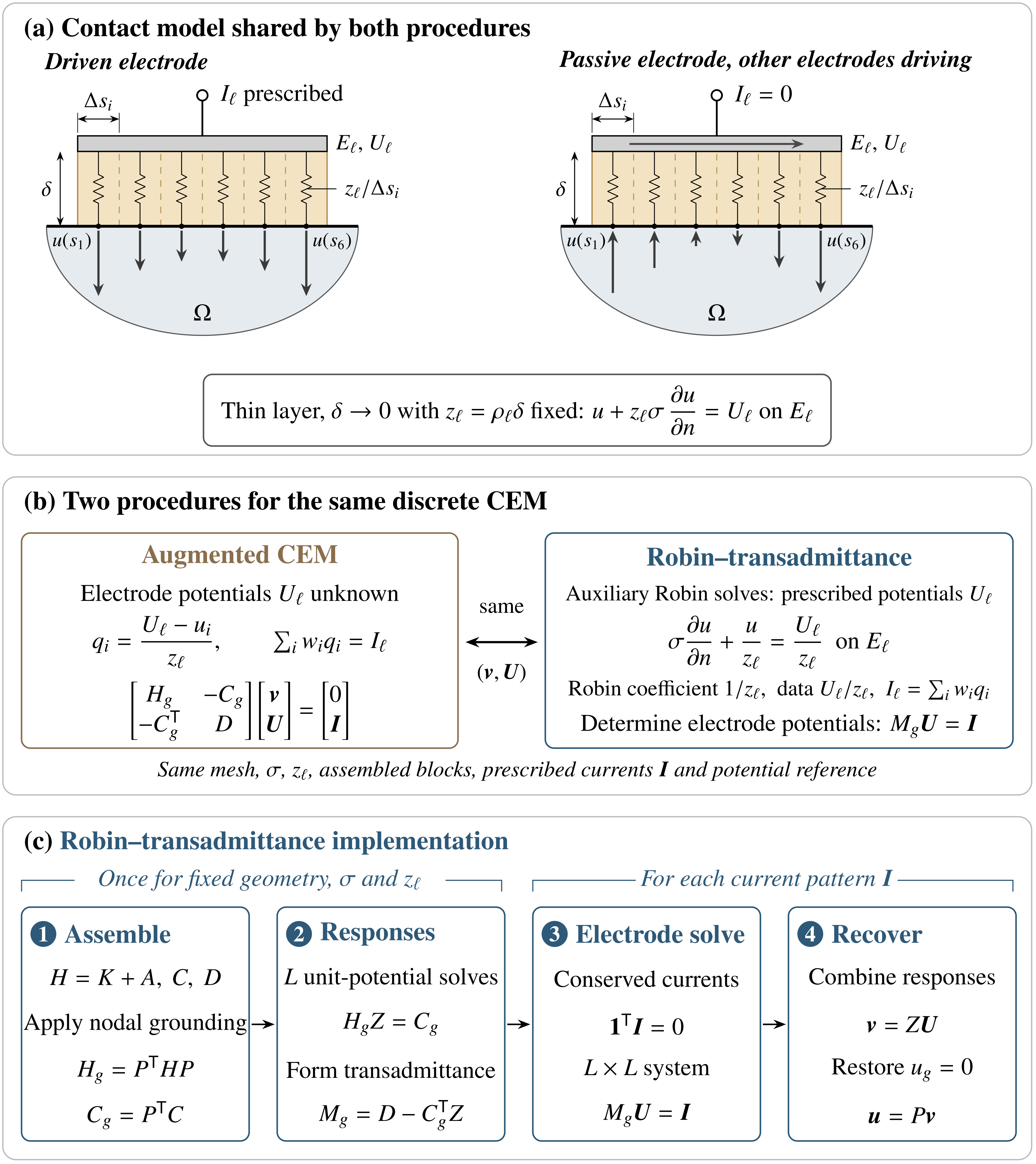}
\caption{Shared contact model and equivalent CEM solution procedures. (a) An equipotential electrode is coupled to local surface potentials through a resistive layer of thickness $\delta$ and resistivity $\rho_\ell$. The limit $\delta\to0$ with $z_\ell=\rho_\ell\delta$ fixed gives the distributed condition~\eqref{eq:contact_zh}. Strips of width $\Delta s_i$ illustrate resistances $z_\ell/\Delta s_i$ per unit out-of-plane thickness ($z_\ell/\Delta A_i$ in three dimensions); they do not represent the finite element mesh. A passive electrode has zero net current but can carry local currents of opposite signs when other electrodes drive current. (b) The augmented and Robin--transadmittance procedures use the same assembled blocks, currents and potential reference. The latter prescribes electrode potentials in auxiliary solves and then enforces the prescribed currents through $M_g\bm U=\bm I$. For P1 boundary nodes, $u_i$ is the nodal potential, $q_i=(U_\ell-u_i)/z_\ell$, and $w_i=\int_{E_\ell}\phi_i\,\mathrm ds$, giving $I_\ell=\sum_iw_iq_i$. (c) Assembly, unit-potential responses, electrode solve and potential recovery. Panels (b) and (c) use nodal grounding, with $\bm u=P\bm v$; section~\ref{sec:grounded_form} gives the details, including the auxiliary ground-row reaction.}\label{fig:graphical_abstract_zh}
\end{figure}

The idea can be understood intuitively as follows. The CEM treats each electrode as a perfect conductor at a single potential $U_\ell$, connected to the body domain $\Omega$ through a distributed contact impedance $z_\ell$. In figure~\ref{fig:graphical_abstract_zh}(a), this interface is drawn as a thin resistive layer whose thickness tends to zero while $z_\ell$ remains fixed, for both a driven and a passive electrode. The local current density need not be uniform and can be concentrated near electrode edges, particularly for small contact impedances \citep{ciulli1996anomalous}. When only the total current is prescribed, the electrode potential $U_\ell$ is unknown and must be chosen so that the current over the whole electrode adds up to the prescribed value. This constraint couples the electrode as a whole, which is why the conventional formulation augments the finite element system.

When the potentials of all electrodes are prescribed, the normal current density entering the body is given locally by $(U_\ell-u)/z_\ell$, which is the Robin contact relation. The current may still concentrate near electrode edges, but the electrode potentials are now known boundary data and the total electrode currents are outputs of the solve. This is why the prescribed-potential problem can be solved with a finite element package that offers Robin boundary conditions.

The next step is to transform the voltage-driven auxiliary solutions into a current-driven solution. The Robin responses to unit electrode potentials give the transadmittance matrix, which maps electrode potentials to total electrode currents. Solving the small electrode system for the prescribed currents gives the electrode potentials, and the nodal potentials follow as the corresponding linear combination of the responses. The nodally grounded implementation used for the EIDORS comparisons computes $L$ Robin responses with the same matrix and solves an $L\times L$ electrode system. If one electrode potential is used as the reference instead, $L-1$ responses and an $(L-1)\times(L-1)$ system suffice (section~\ref{sec:grounded_form}).

Computing the Robin responses can require more work than a direct augmented CEM solve when only a few current patterns are needed. Our goal is to make the CEM accessible through existing finite element packages; their solvers and parallel capabilities may offset this setup cost in suitable applications. Section~\ref{sec:discussion} discusses the computational costs and the conditions for reusing the responses.

\subsection{Complete electrode model}\label{sec:cem}

Let $\Omega\subset\mathbb{R}^{2}$ be a bounded, connected Lipschitz domain, and let the conductivity $\sigma(\boldsymbol x)$ be bounded with a strictly positive lower bound. The electrodes $E_\ell\subset\partial\Omega$, $\ell=1,\ldots,L$ with $L\geq1$, are mutually disjoint and have non-zero lengths. The potential $u$ within the domain and the constant potential $U_\ell$ on each electrode satisfy
\begin{align}
  \nabla\cdot(\sigma\nabla u) &= 0 && \text{in }\Omega, \label{eq:pde_zh}\\
  u+z_\ell\sigma\frac{\partial u}{\partial n} &= U_\ell && \text{on }E_\ell, \label{eq:contact_zh}\\
  \sigma\frac{\partial u}{\partial n} &= 0 && \text{on }\partial\Omega\setminus\bigcup_{\ell=1}^{L}E_\ell, \label{eq:insulation_zh}
\end{align}
where $n$ is the outward unit normal, and $z_\ell>0$ is the contact impedance of electrode $E_\ell$. Positive electrode current denotes injection into the domain; the total electrode current is
\begin{equation}
  I_\ell=\int_{E_\ell}\sigma\frac{\partial u}{\partial n}\,\mathrm{d}s
  =\frac{1}{z_\ell}\left(|E_\ell|U_\ell-\int_{E_\ell}u\,\mathrm{d}s\right).
  \label{eq:electrode_current_zh}
\end{equation}
The current vector $\bm{I}=[I_1,\ldots,I_L]^{\mathsf T}$ satisfies $\bm{1}^{\mathsf T}\bm{I}=0$, and the electrode-potential vector is $\bm{U}=[U_1,\ldots,U_L]^{\mathsf T}$. Because a common additive shift of the domain and electrode potentials leaves the currents unchanged, a potential reference must be fixed; alternative references and the common nodal reference used for the numerical comparisons are described in section~\ref{sec:grounded_form}.

\subsection{Finite element discretisation and block matrices}\label{sec:discretisation}

For a prescribed electrode-potential vector $\bm{U}$, multiplying equation~\eqref{eq:pde_zh} by a test function $\psi\in H^1(\Omega)$, integrating by parts and substituting equations~\eqref{eq:contact_zh} and~\eqref{eq:insulation_zh} gives the weak form of the Robin problem: find $u\in H^1(\Omega)$ such that
\begin{equation}
\int_{\Omega}\sigma\nabla u\cdot\nabla \psi\,\mathrm{d}x
+\sum_{\ell=1}^{L}\frac{1}{z_\ell}\int_{E_\ell}u\psi\,\mathrm{d}s
=\sum_{\ell=1}^{L}\frac{U_\ell}{z_\ell}\int_{E_\ell}\psi\,\mathrm{d}s
\quad\text{for all }\psi\in H^1(\Omega).
\label{eq:weak_robin_zh}
\end{equation}
Let $V_h\subset H^1(\Omega)$ be the continuous, piecewise-linear P1 finite element space on a connected conforming triangulation with nondegenerate triangles, with nodal basis functions $\{\phi_i\}_{i=1}^{N}$ and discrete potential $u_h=\sum_{i=1}^{N}u_i\phi_i$. We distinguish the continuous potential $u$, its finite element approximation $u_h$, and the coefficient vector $\bm u\defeq [u_1,\ldots,u_N]^{\mathsf T}$. Bold symbols denote vectors, while their individual components remain in ordinary type. For multiple current patterns, the same symbols collect the corresponding columns. The approximation gives the linear system
\[
(K+A)\bm u=C\bm U,
\]
where $K,A\in\mathbb R^{N\times N}$ and $C\in\mathbb R^{N\times L}$. We also define the diagonal electrode block $D\in\mathbb R^{L\times L}$. The entries are
\begin{align}
  K_{ij} &\defeq \int_\Omega \sigma\nabla\phi_i\cdot\nabla\phi_j\,\mathrm{d}x, \label{eq:fem_K}\\
  A_{ij} &\defeq \sum_{\ell=1}^{L}\frac{1}{z_\ell}
            \int_{E_\ell}\phi_i\phi_j\,\mathrm{d}s, \\
  C_{i\ell} &\defeq \frac{1}{z_\ell}\int_{E_\ell}\phi_i\,\mathrm{d}s, \\
  D_{\ell\ell} &\defeq \frac{|E_\ell|}{z_\ell}.
  \label{eq:fem_blocks_zh}
\end{align}
When these integrals are evaluated and accumulated exactly, the Robin finite element matrix is defined as $H\defeq K+A$. Here, $K$ is the domain stiffness matrix, $A$ is the Robin boundary matrix, $C$ maps electrode potentials to nodal boundary loads, and the diagonal matrix $D$ depends on electrode lengths and contact impedances. For a straight P1 boundary element of length $h$ on electrode $E_\ell$, the local contributions are
\begin{equation}
  A^{(e)}=\frac{h}{6z_\ell}
  \begin{bmatrix}2&1\\1&2\end{bmatrix},\qquad
  C^{(e)}=\frac{h}{2z_\ell}
  \begin{bmatrix}1\\1\end{bmatrix},\qquad
  D^{(e)}=\frac{h}{z_\ell}.
  \label{eq:local_electrode_blocks_zh}
\end{equation}
These closed-form expressions allow the electrode integrals and their assembly to be checked independently of any finite element package.

The coupled discrete weak form is also useful for identifying the energy: find $(u_h,\bm{U})\in V_h\times\mathbb R^L$ such that
\begin{equation}
\begin{split}
a_h((u_h,\bm{U}),(v_h,\bm{V}))
&\defeq\int_\Omega\sigma\nabla u_h\cdot\nabla v_h\,\mathrm dx\\
&\quad+\sum_\ell\frac1{z_\ell}\int_{E_\ell}(u_h-U_\ell)(v_h-V_\ell)\,\mathrm ds
=\bm{V}^{\mathsf T}\bm{I}
\end{split}
\label{eq:coupled_weak}
\end{equation}
for all $(v_h,\bm{V})\in V_h\times\mathbb R^L$. Testing the nodal and electrode coordinates separately gives the two rows of the block system~\eqref{eq:full_cem}. Conversely, those rows imply the identity for every test pair. In the current-driven CEM, $\bm{U}$ is also unknown. The next section describes the procedures that determine the nodal and electrode potentials and establishes their discrete equivalence.

\subsection{Robin solution and discrete equivalence}\label{sec:equivalence}

We first describe the two solution procedures compared in this study. The conventional approach solves the augmented system for the nodal and electrode potentials simultaneously:
\begin{equation}
F_0\begin{bmatrix}\bm{u}\\\bm{U}\end{bmatrix}=\begin{bmatrix}0\\\bm{I}\end{bmatrix},
\qquad F_0=\begin{bmatrix}H&-C\\-C^{\mathsf T}&D\end{bmatrix}.
\label{eq:full_cem}
\end{equation}
Alternatively, when $H$ is invertible, one can first solve the small electrode system
\begin{equation}
M\bm{U}=\bm{I}\quad\text{for }\bm{U},\qquad
M\defeq D-C^{\mathsf T}H^{-1}C\in\mathbb R^{L\times L},
\label{eq:electrode_system}
\end{equation}
and then solve the standard Robin problem
\begin{equation}
H\bm{u}=C\bm{U}\quad\text{for }\bm{u}.
\label{eq:robin_recovery}
\end{equation}
This procedure follows from a Schur-complement decomposition of equation~\eqref{eq:full_cem}: the first block row gives $\bm{u}=H^{-1}C\bm{U}$, and substituting this into the second block row leaves equation~\eqref{eq:electrode_system}, in which $M$ is the Schur complement of $H$ in $F_0$. The nodal matrix $H\in\mathbb R^{N\times N}$ is large, whereas $M$ has only one row and column per electrode, with $L\ll N$ in practice. For the ungrounded physical CEM, $M$ is singular with nullspace spanned by $\bm1_L$ (corollary~\ref{cor:solvability}), because a common potential shift leaves the currents unchanged; the injected currents must therefore sum to zero, and a potential reference must be fixed before equation~\eqref{eq:electrode_system} is solved. Section~\ref{sec:grounded_form} describes the grounding used here.

To construct $M$ before solving equation~\eqref{eq:electrode_system}, we solve $L$ auxiliary Robin problems with unit electrode potentials, as shown in figure~\ref{fig:graphical_abstract_zh}. Prescribing the electrode potentials $\bm{U}$ turns the first block row into the Robin equation $H\bm{u}=C\bm{U}$, and the resulting electrode currents are $\bm{I}=D\bm{U}-C^{\mathsf T}\bm{u}$. Writing
\begin{equation}
HZ_0=C,\qquad Z_0\defeq H^{-1}C,\qquad
M=D-C^{\mathsf T}Z_0,
\label{eq:ungrounded_responses}
\end{equation}
column $\ell$ of $Z_0$ is the nodal response to unit potential on electrode $\ell$ and zero potential on the others, and column $\ell$ of $M$ contains its total electrode currents. The electrode transadmittance matrix $M$ is therefore assembled from $L$ Robin solves with the same matrix $H$. After equation~\eqref{eq:electrode_system} has given $\bm{U}$, the nodal potentials follow from equation~\eqref{eq:robin_recovery}; when $Z_0$ has been stored, they are simply $\bm{u}=Z_0\bm{U}$, so no further large linear solve is needed. The inverse notation expresses the mathematics: an implementation can factorise $H$ once and solve for the columns without forming $H^{-1}$. The grounded blocks used in the experiments are given in section~\ref{sec:grounded_form}.

A further alternative follows from the other Schur-complement decomposition of equation~\eqref{eq:full_cem}, which eliminates $\bm{U}$ instead: solve
\[
S_u\bm{u}=CD^{-1}\bm{I}\quad\text{for }\bm{u},\qquad
S_u\defeq H-CD^{-1}C^{\mathsf T}\in\mathbb R^{N\times N},
\]
and then $D\bm{U}=\bm{I}+C^{\mathsf T}\bm{u}$ for $\bm{U}$. Here $D$ is positive diagonal, so applying $D^{-1}$ requires only componentwise division. For the exactly assembled energy model of lemma~\ref{lem:structure}, $M$ and $S_u$ are symmetric positive semidefinite with a constant nullspace before a reference is fixed; with the nodal reference of section~\ref{sec:grounded_form}, both corresponding Schur complements are positive definite.

We use the first alternative because its large matrix is $H$ itself. $H$ is the standard Robin finite element matrix, which a general-purpose package assembles through its existing Robin interface, and this allows existing solvers and preconditioners for the Robin operator to be reused. The second alternative instead requires the modified nodal matrix $S_u$, whose term $CD^{-1}C^{\mathsf T}$ couples every pair of nodes on the same electrode. This alternative requires an additional nonlocal electrode-coupling operation beyond standard local Robin assembly, although the coupling, which has rank at most $L$, can be applied without explicitly forming a dense matrix. Both alternatives can reuse their factorisations for new current inputs, so no general speed advantage follows from the choice of elimination order; the advantage of the first alternative is that it uses only standard finite element functionality. We now show that it has the same solution set as the augmented system.

\Needspace{12\baselineskip}
\begin{theorem}[Equivalence of the augmented and Robin solution procedures]\label{thm:equivalence}
Let $H$ be invertible, and use the same blocks $H,C,D$ in equations~\eqref{eq:full_cem} and~\eqref{eq:ungrounded_responses}. For every current vector $\bm{I}$ and pair $(\bm{u},\bm{U})$,
\begin{equation}
F_0\begin{bmatrix}\bm{u}\\\bm{U}\end{bmatrix}=\begin{bmatrix}0\\\bm{I}\end{bmatrix}
\quad\Longleftrightarrow\quad
\bm{u}=Z_0\bm{U}\ \text{and}\ M\bm{U}=\bm{I}.
\label{eq:main_equivalence}
\end{equation}
Thus solving the conventional augmented CEM and solving the electrode system followed by reconstruction from Robin responses give exactly the same solution set.
\end{theorem}
\begin{proof}
\emph{Step 1: expand the block equations.} Equation~\eqref{eq:full_cem} is equivalent to $H\bm{u}=C\bm{U}$ and $D\bm{U}-C^{\mathsf T}\bm{u}=\bm{I}$.

\emph{Step 2: solve the nodal equation.} Since $HZ_0=C$ and $H$ is invertible, $H\bm{u}=C\bm{U}$ holds if and only if $\bm{u}=Z_0\bm{U}$.

\emph{Step 3: substitute into the electrode equation.} This gives $(D-C^{\mathsf T}Z_0)\bm{U}=M\bm{U}=\bm{I}$. Conversely, $M\bm{U}=\bm{I}$ and $\bm{u}=Z_0\bm{U}$ satisfy both original block rows.
\end{proof}
\noindent\emph{Lean 4:} theorem~\ref{thm:equivalence} is \path{augmented_iff_robin} in \path{CEM/Equivalence.lean}. Step 1 is \path{augmented_iff_rows}; step 2 uses \path{response_equation} and \path{nodal_iff_response}; step 3 completes \path{augmented_iff_robin}.

This block-elimination result requires neither invertibility of $M$ nor symmetry nor current conservation. The unit potentials in the auxiliary solves only generate the columns of $Z_0$ and $M$; the actual electrode potentials are unknowns determined by $M\bm{U}=\bm{I}$. The two-electrode example in supplementary section~10.4 illustrates this distinction. The following energy result justifies $H$ invertibility for the CEM.

\begin{lemma}[Structure of the exact discrete CEM]\label{lem:structure}
For the discretisation of section~\ref{sec:discretisation} under the assumptions of section~\ref{sec:cem}, exact assembly gives
\[
K\bm1_N=0,\qquad A\bm1_N=C\bm1_L,\qquad
C^{\mathsf T}\bm1_N=D\bm1_L.
\]
The matrix $H$ is symmetric positive definite. The full matrix $F_0$ is symmetric positive semidefinite, with nullspace consisting exactly of common constant potentials.
\end{lemma}
\begin{proof}
The block identities follow from $\sum_i\phi_i=1$ and equations~\eqref{eq:fem_K}--\eqref{eq:fem_blocks_zh}. Expanding the quadratic form gives, for arbitrary $\bm{u}$ and $\bm{U}$,
\begin{equation}
\begin{bmatrix}\bm{u}\\\bm{U}\end{bmatrix}^{\!\mathsf T}F_0
\begin{bmatrix}\bm{u}\\\bm{U}\end{bmatrix}
=\int_\Omega\sigma|\nabla u_h|^2\,\mathrm dx+
\sum_{\ell=1}^L\frac1{z_\ell}\int_{E_\ell}(u_h-U_\ell)^2\,\mathrm ds.
\label{eq:energy}
\end{equation}
All terms are non-negative. Zero volume energy forces $u_h$ to be constant on each triangle; continuity and mesh connectivity make it one constant throughout the domain. Since every electrode has positive length, zero contact energy forces each $U_\ell$ to equal that constant. Conversely, common constants satisfy $F_0\bm1=0$. Thus $F_0$ is positive semidefinite with exactly this nullspace. Setting $\bm{U}=0$ makes the constant zero, so $\bm{u}^{\mathsf T}H\bm{u}=0$ only for $\bm{u}=0$, proving $H$ positive definite. Supplementary section~10.1 gives the corresponding finite weighted-square proof and its explicit P1 coefficients.
\end{proof}
\noindent\emph{Lean 4:} In \path{CEM/Assembly.lean}, \path{DiscreteModel.weak_form_iff} links the finite weak form to the blocks; \path{K_constant}, \path{A_constant}, \path{coupling_constant}, \path{energy_zero_iff}, \path{full_kernel_iff} and \path{H_posDef} in that namespace follow the lemma.

\begin{corollary}[Current compatibility and uniqueness up to a constant]\label{cor:solvability}
For the exactly assembled model, $Z_0\bm1_L=\bm1_N$ and $M\bm1_L=0$. The matrix $M$ is symmetric positive semidefinite with $\ker M=\operatorname{span}\{\bm1_L\}$. Every current vector with $\bm1_L^{\mathsf T}\bm{I}=0$ therefore has a full CEM solution, unique up to a common additive constant in $\bm{u}$ and $\bm{U}$.
\end{corollary}
\begin{proof}
The block identities of lemma~\ref{lem:structure} give $H\bm1_N=C\bm1_L$, hence $Z_0\bm1_L=\bm1_N$ and $M\bm1_L=0$. Symmetry follows from $M=D-C^{\mathsf T}H^{-1}C$. For $\bm{u}=Z_0\bm{U}$, the quadratic form in equation~\eqref{eq:energy} equals $\bm{U}^{\mathsf T}M\bm{U}$, so $M$ is positive semidefinite and has exactly the constant nullspace. Its restriction to $\bm1_L^\perp$ is consequently positive definite and invertible. Every balanced $\bm{I}$ has a unique solution there; all other solutions differ by a constant. Applying theorem~\ref{thm:equivalence} reconstructs $\bm{u}$, and $Z_0\bm1_L=\bm1_N$ gives the same nodal shift.
\end{proof}
\noindent\emph{Lean 4:} \path{DiscreteModel.assembled_current_compatibility} in \path{CEM/MeanZero.lean} collects the complete conclusion, including existence and the common-shift characterisation of all solutions. It combines the matrix properties in \path{CEM/Schur.lean} with \path{assembled_zeroMean_solution}, whose construction is given in supplementary proposition~S1.

A potential reference selects one solution from this family; it does not establish the equivalence in theorem~\ref{thm:equivalence}.

\subsection{Potential references and the EIDORS comparison}\label{sec:grounded_form}
One may impose a zero-mean electrode potential, fix one electrode potential, or fix one nodal potential. The first two choices retain the ungrounded Robin matrix $H$ and response matrix $Z_0$. A zero-mean reference restricts $M\bm{U}=\bm{I}$ to $\bm1_L^\perp$, with the explicit formula given below. Fixing $U_r=0$ gives an $(L-1)\times(L-1)$ system by deleting electrode row and column $r$; the following result recovers the missing equation.

\begin{lemma}[Recovery of an omitted reference equation]\label{lem:reference_row}
Let $M$ be symmetric with $M\bm1_L=0$, and let $\bm1_L^{\mathsf T}\bm{I}=0$. If $U_r=0$ and $(M\bm{U})_\ell=I_\ell$ for every $\ell\ne r$, then $(M\bm{U})_r=I_r$. Thus the reduced system and the full equation with $U_r=0$ are equivalent, even when the prescribed current $I_r$ is non-zero.
\end{lemma}
\begin{proof}
Symmetry gives $\bm1_L^{\mathsf T}M=0$. Therefore the sum of the residuals $M\bm{U}-\bm{I}$ is zero; when all retained residuals vanish, the omitted residual must vanish too. Conversely, restriction of a full solution gives the reduced system.
\end{proof}
\noindent\emph{Lean 4:} lemma~\ref{lem:reference_row} is \path{reference_reduction_iff_full} in \path{CEM/Grounding.lean}. It uses \path{total_current_zero} and \path{recover_ground_row}, and allows a non-zero omitted current $I_r$.

For numerical comparisons we adopt the nodal reference used by the EIDORS forward solver, and impose it at the same node in every implementation. Let $P$ insert a zero at node $g$, so $\bm{u}=P\bm{v}$, and define $H_g\defeq P^{\mathsf T}HP$ and $C_g\defeq P^{\mathsf T}C$. Deleting that row and column gives
\begin{equation}
F_g\begin{bmatrix}\bm{v}\\\bm{U}\end{bmatrix}
=\begin{bmatrix}H_g&-C_g\\-C_g^{\mathsf T}&D\end{bmatrix}
\begin{bmatrix}\bm{v}\\\bm{U}\end{bmatrix}=\begin{bmatrix}0\\\bm{I}\end{bmatrix}.
\label{eq:grounded}
\end{equation}
Applying the same block elimination as in theorem~\ref{thm:equivalence}, now to $H_g,C_g,D$, gives
\begin{equation}
H_gZ=C_g,\qquad M_g=D-C_g^{\mathsf T}Z,
\label{eq:responses}
\end{equation}
and the grounded Robin computation is
\begin{equation}
M_g\bm{U}=\bm{I},\qquad \bm{v}=Z\bm{U},\qquad \bm{u}=P\bm{v}.
\label{eq:robin}
\end{equation}
Here $Z$ and $M_g$ differ from $Z_0$ and $M$: the auxiliary solves also fix $u_g=0$. Nodal grounding is a constraint on a discrete degree of freedom, not a pointwise condition in the continuous weak formulation.

\begin{proposition}[Grounded implementation of the same CEM]\label{prop:grounding}
Suppose that $F_0$ is symmetric, $F_0\bm1=0$, its ground-deleted submatrix is $F_g$, and $\bm1^{\mathsf T}\bm{I}=0$. Then a pair $(\bm{v},\bm{U})$ satisfies equation~\eqref{eq:grounded} if and only if $(P\bm{v},\bm{U})$ satisfies every row of equation~\eqref{eq:full_cem}. If $H_g$ is invertible, these conditions are also equivalent to equation~\eqref{eq:robin}. If $M_g$ is additionally invertible, the solution with $u_g=0$ exists uniquely and is given by $\bm{U}=M_g^{-1}\bm{I}$, $\bm{u}=PZM_g^{-1}\bm{I}$.
\end{proposition}
\begin{proof}
Inserting a zero ground value identifies the retained full equations with equation~\eqref{eq:grounded}. Symmetry and $F_0\bm1=0$ give $\bm1^{\mathsf T}F_0=0$. Since the right-hand side also sums to zero, satisfaction of the retained rows forces the deleted row to hold. The reverse implication follows by restriction. The remaining statements follow from theorem~\ref{thm:equivalence} applied to the grounded blocks and, for uniqueness, from invertibility of $M_g$.
\end{proof}

\noindent\emph{Lean 4:} In \path{CEM/Complete.lean}, \path{full_cem_iff_robin}, \path{full_cem_iff_explicit_solution} and \path{full_cem_exists_unique} give the equivalence, explicit solution and uniqueness in proposition~\ref{prop:grounding}. Its row-recovery step is \path{grounded_reduction_iff_full} in \path{CEM/Grounding.lean}; the reduced-system helpers \path{augmented_iff_explicit_solution} and \path{grounded_exists_unique} are in \path{CEM/Equivalence.lean}. The finite-model result \path{DiscreteModel.assembled_grounded_full_exists_unique} in \path{CEM/GroundedEnergy.lean} derives the inverse hypotheses.

For the exactly assembled model, equation~\eqref{eq:energy} restricted to $u_g=0$ proves that $F_g$ is positive definite: its only possible zero-energy vector is a common constant, and the fixed nodal value makes that constant zero. Its principal block $H_g$ and Schur complement $M_g$ are therefore positive definite. Because $F_g$ is positive definite, its other Schur complement $S_{u,g}\defeq H_g-C_gD^{-1}C_g^{\mathsf T}$, taken with respect to the diagonal block $D$, is positive definite as well.

An individual grounded auxiliary response can have non-zero total electrode current. Writing $\bm{I}_{\mathrm{aux}}$ for column $\ell$ of $M_g$, the deleted row supplies a reaction $r_g=-\bm1^{\mathsf T}\bm{I}_{\mathrm{aux}}$. The result in proposition~\ref{prop:grounding} ensures that this reaction vanishes in the final balanced combination. Each ungrounded auxiliary response already satisfies every source-free nodal equation.

\noindent\emph{Lean 4:} \path{ground_reaction_eq_neg_sum} in \path{CEM/Grounding.lean} checks the auxiliary reaction identity.

The result in corollary~\ref{cor:solvability} establishes uniqueness up to a constant for the exactly assembled energy model. The next result instead uses only the matrix hypotheses of proposition~\ref{prop:grounding}, without an energy assumption.

\begin{corollary}[Independence of voltage differences from the reference]\label{cor:gauge}
Under the hypotheses of proposition~\ref{prop:grounding}, with $H_g$ and $M_g$ invertible, any two full solutions for the same current input differ by a common constant. Thus nodal grounding, a fixed electrode potential and a zero-mean electrode reference yield identical voltage differences in exact arithmetic.
\end{corollary}
\begin{proof}
For any two full solutions $\bm{x}$ and $\bm{y}$, subtract their respective ground-node values to obtain $\bm{x}^g=\bm{x}-x_g\bm1$ and $\bm{y}^g=\bm{y}-y_g\bm1$. The identity $F_0\bm1=0$ preserves the equations. Uniqueness at the common nodal reference gives $\bm{x}^g=\bm{y}^g$, hence $\bm{x}-\bm{y}=(x_g-y_g)\bm1$ and $x_i-x_j=y_i-y_j$ for every pair of indices $i,j$.
\end{proof}
\noindent\emph{Lean 4:} \path{full_cem_solutions_equal_after_grounding} and \path{full_cem_solutions_same_voltage_differences} in \path{CEM/Gauge.lean}.

For the exactly assembled model, let $J\defeq\bm1_L\bm1_L^{\mathsf T}/L$, the orthogonal projector onto constants. Since $M$ is positive semidefinite with $\ker M=\operatorname{span}\{\bm1_L\}$ (corollary~\ref{cor:solvability}), $M+J$ is positive definite and hence invertible. The zero-mean electrode reference is obtained from
\begin{equation}
J=\frac1L\bm1_L\bm1_L^{\mathsf T},\qquad
G\defeq(M+J)^{-1}-J=M^\dagger.
\label{eq:cem_pseudoinverse}
\end{equation}
Here $M$ is the ungrounded electrode matrix, and $M^\dagger$ denotes its Moore--Penrose pseudoinverse. For balanced currents, $\bm{U}=G\bm{I}$ and $\bm{u}=Z_0\bm{U}$ give the unique zero-mean solution. Supplementary proposition~S1 proves this construction and its uniqueness, with the matching Lean statements. Equivalently, subtract the electrode mean from both the nodal and electrode potentials of any full solution. In the error analysis of section~\ref{sec:rational_results}, expressing both computed and reference electrode outputs in this convention removes the mean component of their difference. This is distinct from re-solving a separately assembled floating-point system with a zero-mean constraint (section~\ref{sec:discussion}).

Reference invariance concerns exact solutions of the same discrete problem. For the numerical comparisons, the mesh, inputs, signs and nodal reference are aligned with EIDORS before comparing solution procedures and assembly effects. A common return electrode in a current pattern is not a potential reference.

\subsection{Correspondence with the Lean 4 proofs}\label{sec:lean}
The formalisation uses Lean~4.24.0 \citep{demoura2021lean4} and mathlib \citep{mathlib2020} at two connected levels. The generic equivalence and reference results hold over a commutative ring. Over exact real numbers, the finite energy model constructs the blocks from gradient and boundary maps, proves the coefficient weak form and lemma~\ref{lem:structure}, derives the properties of $M,F_g,H_g,M_g$, and constructs the zero-mean solution and pseudoinverse. In particular, \path{DiscreteModel.assembled_iff_robin} derives $H$ invertibility before applying theorem~\ref{thm:equivalence}. The P1 boundary module verifies the local coefficients in equation~\eqref{eq:local_electrode_blocks_zh}. Table~\ref{tab:lean} follows the manuscript's result order.

The passage from continuous integrals and mesh geometry to the finite input data remains a paper argument; continuous trace theory, platform assembly programs, experimental reference coefficients, floating-point solvers and PDE convergence are outside the certificate. For shared stored coefficients treated as exact numbers, theorem~\ref{thm:equivalence} still applies when $H$ is invertible; the reference results additionally require their exact matrix identities. Run \texttt{lake build}, then \texttt{lake env lean Paper.lean}, to check the proofs and display their statements in paper order. Supplementary sections~10.1--10.3 explain the finite inputs, Lean notation and pseudoinverse proof; section~10.4 supplies an exact example.

\begin{table}[!htb]\centering\small
\caption{Paper--Lean correspondence in reading order. Names are in namespace \texttt{CEM}; paths are relative to the proof package. Detailed proof steps appear in supplementary table~S29.}\label{tab:lean}
\def\UrlBreaks{\do\/\do\_\do\.}
\begin{tabularx}{\linewidth}{@{}>{\raggedright\arraybackslash}p{0.23\linewidth}Y Y@{}}\toprule
Paper result & Starting data or hypotheses & Lean statement and file\\\midrule
theorem~\ref{thm:equivalence}: augmented/Robin equivalence & Identical $H,C,D$; $H$ invertible & \path{augmented_iff_robin}\newline\path{CEM/Equivalence.lean}\\\addlinespace[3pt]
lemma~\ref{lem:structure}: energy and definiteness & Connected P1 mesh; positive conductivity and contact impedances; electrodes of positive length (sections~\ref{sec:cem}--\ref{sec:discretisation}); finite inputs in supplementary section~10.1 & \path{DiscreteModel.energy_zero_iff}, \path{DiscreteModel.H_posDef}\newline\path{CEM/Assembly.lean}\\\addlinespace[3pt]
corollary~\ref{cor:solvability}: compatibility and uniqueness & lemma~\ref{lem:structure}; balanced current for existence & \path{DiscreteModel.assembled_current_compatibility}\newline\path{CEM/MeanZero.lean}\\\addlinespace[3pt]
lemma~\ref{lem:reference_row}: omitted reference equation & Symmetry, constant null vector, balanced full load & \path{reference_reduction_iff_full}\newline\path{CEM/Grounding.lean}\\\addlinespace[3pt]
proposition~\ref{prop:grounding}: nodal implementation & Stated matrix hypotheses and balanced current & \path{full_cem_iff_robin}, \path{full_cem_exists_unique}\newline\path{CEM/Complete.lean}\\\addlinespace[3pt]
corollary~\ref{cor:gauge}: reference invariance & Uniqueness at the common nodal reference & \path{full_cem_solutions_same_voltage_differences}\newline\path{CEM/Gauge.lean}\\\addlinespace[3pt]
Supplementary proposition S1: zero mean and pseudoinverse & Derived positivity and constant energy kernel of $M$ & \path{cemInverse_penrose}, \path{cemInverse_unique}, \path{zeroMean_solution_iff}\newline\path{CEM/MeanZero.lean}\\
\bottomrule
\end{tabularx}
\end{table}

\section{Numerical experiments and evaluation methods}\label{sec:numerical}
\subsection{Experimental design and comparison groups}\label{sec:design}
Two complementary test cases are used. The 2692-node case compares the platform-native Robin implementations with EIDORS and, through shared matrices and recorded local integrals, separates the effects of the solution and assembly procedures. The 2535-node case is assembled and solved exactly in rational arithmetic and provides a common reference for the numerical error of each implementation before and after iterative refinement. Both cases use 16 electrodes and P1 triangular elements but differ in mesh and in the representation of the material interface, so they define distinct discrete problems. A 1055-node rational case provides an additional exact-reference comparison (supplementary table~S27). Parameter variations on the 2692-node mesh and an analytical solution on a rectangular domain provide further checks.

The software benchmark $R$ is the conventional augmented CEM output of the EIDORS forward function \texttt{fwd\_solve}. The platform-native Robin--transadmittance implementation on platform $s$ is denoted by $P_s$. Three further implementations, $A$, $B_s$ and $C_s$, are used to examine the effects of the solution formulation, the platform-specific solution procedure and the assembly procedure, respectively (table~\ref{tab:paths}). In the 2535-node case, the exact discrete reference $\bm{x}_{\mathbb Q}$ is common to $R$, $A$ and the five $P_s$ implementations, so EIDORS is evaluated in the same way as the other implementations.

The output of each implementation's assembly and double-precision solution is denoted by $\bm{x}_s=\bm{x}_s^{(0)}$ and called the directly computed result. The iterative refinement of section~\ref{sec:refinement_method} gives $\bm{x}_s^{(k)}$ after step $k$ and $\bm{x}_s^\star=\bm{x}_s^{(4)}$ after four steps.

\begin{table}[!htb]\centering\small
\caption{Implementations and sources of assembled finite element matrices. The subscript $s$ identifies the platform, and all implementations impose the same nodal grounding constraint. The implementation label $A$ is unrelated to the Robin boundary matrix $A$.}\label{tab:paths}
\begin{tabularx}{\linewidth}{@{}l Y Y@{}}\toprule
Implementation & Source of assembled matrices & Procedure and purpose\\\midrule
$R$ & Matrices assembled in the EIDORS benchmark & Solve the conventional augmented CEM through the EIDORS forward interface; provide the software benchmark\\
$P_s$ & The platform's own finite element integration and assembly interfaces & Assemble and solve the Robin--transadmittance CEM to test cross-platform implementation\\
$A$ & $H_g,C_g,D$ extracted from the matrices used by $R$ & Solve the Robin--transadmittance system in MATLAB; compare solution formulations with $R$\\
$B_s$ & Benchmark matrices numerically identical to those used by $A$ & Apply platform $s$'s solution procedure to compare solutions with fixed assembled coefficients\\
$C_s$ & Platform-computed local integrals and their accumulated contributions & Retain local element and electrode-edge matrices; compare with $B_s$ to assess the effects of assembly changes\\
\bottomrule
\end{tabularx}
\end{table}

\subsection{Discrete inputs for the two test cases}\label{sec:meshes}
All numerical test cases are dimensionless, including both disc models, the parameter variations and the rectangular analytical check. Coordinates and lengths, conductivity $\sigma$, contact impedance $z_\ell$, prescribed currents and computed potentials are expressed in the same nondimensional model convention; no physical length, current or voltage unit is assigned. All absolute errors, residuals, tables and figure values below and in the supplementary material use this convention. Relative errors and condition numbers are dimensionless by definition.

\noindent\textbf{2692-node cross-platform case.}
This normalised model occupies the unit disc, whose boundary is approximated by 256 straight chords. The mesh has 2692 nodes and 5126 P1 triangles. Its 16 electrodes are numbered anticlockwise, with their centres spaced at angular intervals of $22.5^\circ$; electrodes and inter-electrode gaps both have a nominal angular width of $11.25^\circ$, so each electrode spans eight boundary chords. The background conductivity is $\sigma_{\mathrm{bg}}=1$, and a circular inclusion of radius $0.22$ and conductivity $\sigma_{\mathrm{inc}}=2$, centred at $(0.30,0.15)$, has an interface represented by 64 straight segments. All contact impedances are $z_\ell=0.125$. The mesh was generated with the Frontal-Delaunay algorithm in Gmsh 4.14.0 and has a minimum angle of $23.60^\circ$. All implementations use the same nodal coordinates, element connectivity, material labels and electrode labels. Figure~\ref{fig:circle16} shows the mesh and the nodal potentials computed by EIDORS.

\begin{figure}[!htbp]\centering
\includegraphics[width=\linewidth]{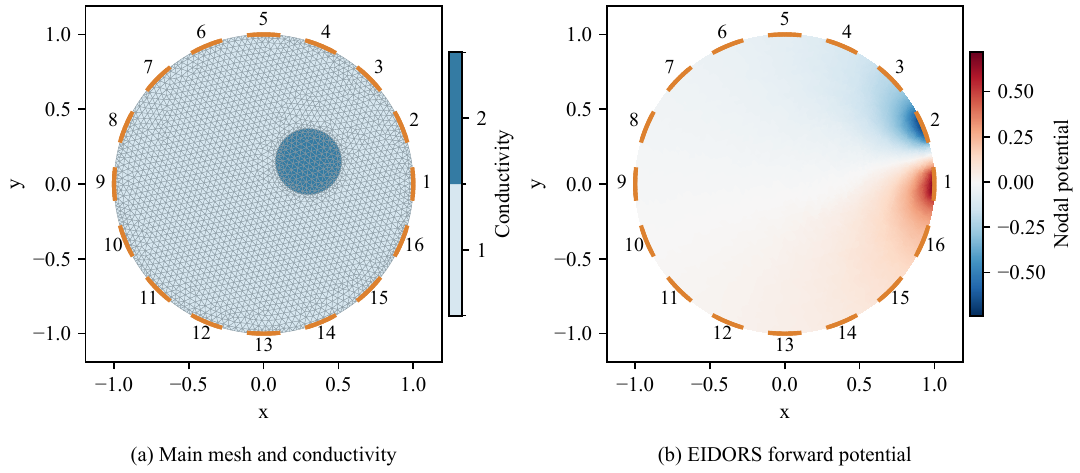}
\caption{2692-node cross-platform case. (a) Common mesh and elementwise conductivity. (b) Directly computed EIDORS nodal potentials for the adjacent stimulation $\bm{e}_1-\bm{e}_2$, displayed using P1 interpolation. Orange segments mark the electrode chords.}\label{fig:circle16}
\end{figure}

Because platforms may reorder their internal indices, the imported coordinates and degree-of-freedom mappings are matched to the common ordering, and all assembled matrices are compared in that ordering. The potential at node 2692 is fixed to zero. With $\bm{e}_i$ denoting the electrode vector whose $i$th component is 1 and whose other components are zero, the current inputs are $\bm{e}_1-\bm{e}_2$ and $\bm{e}_i-\bm{e}_{16}$ ($i=1,\ldots,15$), giving 16 columns of rank 15. The last 15 columns form a basis of the current-conserving subspace, and $\bm{e}_1-\bm{e}_2$ is included as an example of adjacent stimulation. Electrode 16 is the common return electrode of these columns rather than a potential reference, so $U_{16}$ is not set to zero. With the zero ground row restored, each full output contains $(2692+16)\times16=43328$ potential values.

\noindent\textbf{2535-node exact discrete benchmark case.}\label{sec:rational_setup}
This mesh has 2535 nodes, 4940 P1 triangles and 16 electrodes on a disc-like polygon with 128 straight boundary segments, four per electrode. Nodal coordinates have the rational form $p/65536$ and are exactly representable in double precision, and the electrode edges are constructed from integer Pythagorean triples, so that all 64 electrode-edge lengths are rational; interior edge lengths need not be rational. The boundary nodes lie close to, but not exactly on, the unit circle. The minimum mesh angle is $28.12^\circ$. Elements whose centroids lie within a distance of $0.22$ of $(0.30,0.15)$ have conductivity $\sigma_{\mathrm{inc}}=2$ and the remaining elements $\sigma_{\mathrm{bg}}=1$, so the material interface follows element edges. All contact impedances are $z_\ell=1/8$, the current patterns are those of the 2692-node case, and node 2535, at the origin, is grounded. Figure~\ref{fig:rational_mesh} shows the mesh and the exact discrete reference potentials.

\begin{figure}[!htbp]\centering
\includegraphics[width=\linewidth]{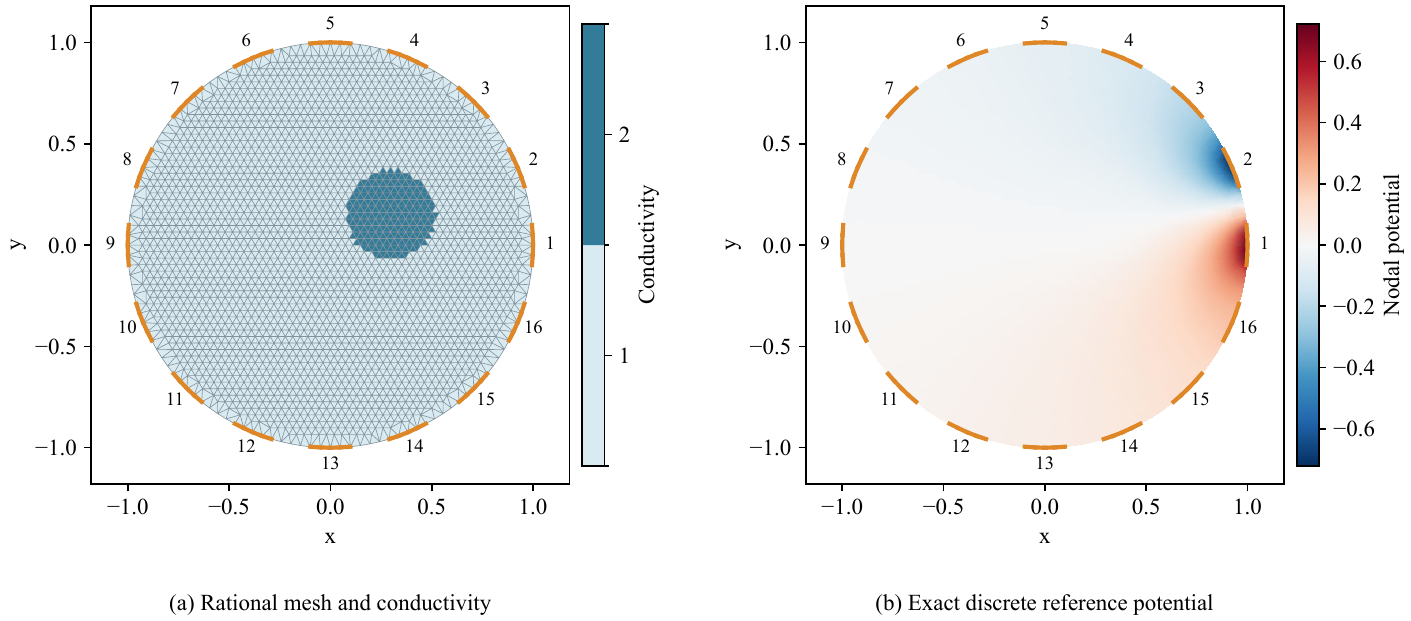}
\caption{2535-node exact discrete benchmark case. (a) Mesh, electrodes and elementwise conductivity. (b) Exact discrete reference nodal potentials for the adjacent stimulation $\bm{e}_1-\bm{e}_2$, displayed using P1 interpolation. Orange segments mark the electrode edges.}\label{fig:rational_mesh}
\end{figure}

The exact discrete reference is the solution of the discrete CEM obtained by evaluating and assembling the finite element integrals and solving the resulting systems entirely in rational arithmetic. It therefore separates numerical errors in assembly and solution from discretisation error; a second independent exact assembly gives the same matrices and solution, with zero residual in the full discrete equations. Supplementary section~9 describes the construction, exact checks and correctly rounded double-precision representations of all $(2535+16)\times16=40816$ reference values.

\subsection{Finite element implementations and shared-matrix comparisons}\label{sec:implementations}
The EIDORS benchmark uses MATLAB R2023b Update 4 and EIDORS 3.12 with the first-order forward solver and system-matrix assembly. The forward solution is requested for all nodal and CEM electrode degrees of freedom, so the returned array contains the $N$ nodal potentials followed by the $L$ electrode potentials; with an $L\times L$ identity measurement matrix, the measurements reproduce $\bm{U}$ after reordering. The benchmark $R$ is this interface output (supplementary section~6.1).

The five platforms are DOLFINx, FreeFEM, GetFEM, MFEM and NGSolve \citep{baratta2025dolfinx,hecht2012freefem,renard2020getfem,anderson2021mfem,schoeberl2014ngsolve}. Software versions and solver configurations are reported in supplementary section~6.1. In $P_s$, each platform evaluates the domain diffusion term and the Robin terms over entire electrodes with its own integration routines and accumulates the local contributions according to mesh connectivity to form $H$. Column $\ell$ of $C$ is the load vector for unit potential on electrode $\ell$, and $D_{\ell\ell}$ is the integral of $z_\ell^{-1}$ over that electrode, which MFEM obtains by applying its boundary linear form to the all-ones vector. We refer to this procedure as platform-native assembly, in contrast to importing the benchmark matrices. Supplementary section~6.2 gives the boundary terms, quadrature rules and interfaces used on each platform.

In the shared-matrix comparisons, $H_g$, $C_g$ and $D$ are taken entry by entry from the finite element matrix used by EIDORS, so that the implementations compared share the stored double-precision coefficients and not merely the mesh and parameters; in particular, $H$ is the exported EIDORS block rather than a separately computed $K+A$. The augmented formulation solves the full $F_g$ system, whereas the Robin--transadmittance formulation follows equations~\eqref{eq:responses}--\eqref{eq:robin}. Implementation $A$ is our MATLAB Robin--transadmittance solver based on LU factorisation. For multiple right-hand sides, the EIDORS linear solver first computes a QR factorisation of the right-hand-side matrix, solves the system with the orthonormal factor as the right-hand side and multiplies the result by the triangular factor, so its operation sequence differs from that of $A$.

Within each platform, $P_s$, $B_s$ and $C_s$ use the same linear algebra interfaces, so that only the source of the assembled matrices changes between them. NGSolve uses its own matrix-inversion and matrix-product operations, whereas DOLFINx, FreeFEM, GetFEM and MFEM solve the response equations with PETSc dense LAPACK LU, UMFPACK, MUMPS and LAPACK-enabled dense matrix inversion, respectively; on these four platforms, NumPy forms the electrode transadmittance matrix and recovers the nodal potentials. All direct solutions use real double-precision (binary64) arithmetic in a single-threaded configuration. In both 16-electrode cases, every implementation was run twice from the same fixed inputs, and the two runs agreed bitwise (supplementary section~6.3).

The local-integral-recording implementation $C_s$ retains the local matrices of every element and electrode edge of the 2692-node mesh, and part of the accumulation is carried out in a common order, so that differences in local integration can be separated from those introduced by subsequent accumulation (supplementary section~4). Because $C_s$ and $P_s$ use different assembly procedures, their potential differences combine the assembly changes with the resulting changes in the floating-point solution (supplementary section~7). Supplementary section~3 decomposes the differences between $C_s$ and $R$ componentwise.

\subsection{Reference solutions, system residuals and error measures}\label{sec:metrics}
The full potential output is $\bm{x}_s=[\bm{u}_s^{\mathsf T},\bm{U}_s^{\mathsf T}]^{\mathsf T}$, with the ground row restored to exact zero. Deleting the ground degree of freedom gives $\bm{x}_{g,s}=[\bm{v}_s^{\mathsf T},\bm{U}_s^{\mathsf T}]^{\mathsf T}$. In the error expressions below, $s$ identifies the implementation under evaluation: $R$, $A$, $B_s$, $C_s$ or $P_s$. We denote a stored floating-point assembled matrix by $\widehat F_{g,s}$ to distinguish it from an exactly assembled matrix. Three solutions are relevant: the exact discrete reference $\bm{x}_{\mathbb Q}$, obtained by exact assembly and solution from fixed inputs; the exact solution $\widetilde{\bm{x}}_s$ of the system whose coefficients are the stored entries of $\widehat F_{g,s}$ treated as exact; and the approximate solution $\bm{x}_s^{(k)}$ computed directly or after finitely many refinement steps. Specifically,
\begin{equation}
F_{g,\mathbb Q}\bm{x}_{g,\mathbb Q}=\bm{b},\qquad
\widehat F_{g,s}\widetilde{\bm{x}}_{g,s}=\bm{b},\qquad
\bm{b}=[0^{\mathsf T},\bm{I}^{\mathsf T}]^{\mathsf T}.
\label{eq:reference_systems}
\end{equation}
The exact solution $\widetilde{\bm{x}}_s$ of the stored assembled system defines the two terms in the following error decomposition:
\begin{equation}
\bm{x}_s^{(k)}-\bm{x}_{\mathbb Q}
=\underbrace{\bm{x}_s^{(k)}-\widetilde{\bm{x}}_s}_{\substack{\text{solution error for the}\\\text{fixed assembled system}}}
+\underbrace{\widetilde{\bm{x}}_s-\bm{x}_{\mathbb Q}}_{\substack{\text{assembly-induced}\\\text{solution shift}}}.
\label{eq:exact_error_split}
\end{equation}
Here $\widehat F_{g,s}$ is assumed to be invertible. The first term is the solution error for the fixed assembled system; the second, the assembly-induced solution shift, is unaffected by refinement. The two terms may reinforce or partially cancel each other, so reducing the first need not reduce their sum. Equation~\eqref{eq:exact_error_split} concerns numerical error within one discrete problem and does not include the discretisation error relative to the continuous CEM.

The grounded augmented-system residual of each implementation is evaluated with its own assembled matrix and includes both the nodal and the electrode equations, but not the deleted ground row:
\begin{equation}
\bm{r}_s^{(k)}=\bm{b}-\widehat F_{g,s}\bm{x}_{g,s}^{(k)},\qquad
\rho_s^{(k)}=\|\bm{r}_s^{(k)}\|_{\max},
\label{eq:full_residual}
\end{equation}
where, for multiple stimulation columns, $\|V\|_{\max}=\max_{i,j}|V_{ij}|$ is the maximum absolute entry over all columns and relevant degrees of freedom. The residual measures how well the assembled equations are satisfied. For a fixed assembled system, the forward error is $\bm{x}_{g,s}^{(k)}-\widetilde{\bm{x}}_{g,s}=-\widehat F_{g,s}^{-1}\bm{r}_s^{(k)}$, so even a decreasing residual need not give a monotonically decreasing forward error.

For each stimulation column, the normalised backward error is calculated as
\begin{equation}
\eta_{\infty,s}
=\frac{\|\bm{b}-\widehat F_{g,s}\bm{x}_{g,s}\|_\infty}
{\|\widehat F_{g,s}\|_\infty\|\bm{x}_{g,s}\|_\infty+\|\bm{b}\|_\infty},
\label{eq:backward_error}
\end{equation}
and we report the maximum over all stimulation columns. The numerator is evaluated in exact rational arithmetic, and the denominator in double precision. For $R$, $A$ and $B_s$, the backward error is evaluated with the benchmark assembled system, and for $C_s$ and $P_s$ with their own assembled systems. For $\bm{I}_{\mathrm{calc}}=D\bm{U}-C_g^{\mathsf T}\bm{v}$, the discrepancy $\bm{I}-\bm{I}_{\mathrm{calc}}$ is exactly the electrode-equation part of the residual~\eqref{eq:full_residual} and therefore does not provide an independent check of the boundary-flux integration.

For the 2692-node case, we report maximum absolute and relative differences from EIDORS; the relative electrode-potential difference is $e_U=\|\bm{U}_{P_s}-\bm{U}_R\|_F/\|\bm{U}_R\|_F$, and the relative nodal-potential difference is defined analogously. For the 2535-node case, we evaluate errors relative to the exact discrete reference, hereafter termed reference errors:
\begin{equation}
E_{U,s}^{(k)}=\|\bm{U}_s^{(k)}-\bm{U}_{\mathbb Q}\|_{\max},\qquad
E_{u,s}^{(k)}=\|\bm{u}_s^{(k)}-\bm{u}_{\mathbb Q}\|_{\max}.
\label{eq:reference_errors}
\end{equation}
These errors are computed exactly from the unrounded reference. To separate voltage-difference errors from a common potential shift, we also evaluate $E_{\Delta U,s}^{(k)}=\|T_{16}(\bm{U}_s^{(k)}-\bm{U}_{\mathbb Q})\|_{\max}$, where $(T_{16}\bm{U})_i=U_i-U_{16}$, $i=1,\ldots,15$. Because $T_{16}\bm1=0$, this measure is invariant under a common potential shift. Adjacent-electrode voltage-difference errors are evaluated analogously.

\subsection{Iterative refinement with fixed assembled matrices}\label{sec:refinement_method}
Iterative refinement is used here as a diagnostic of numerical error. All stored finite element matrix entries are kept fixed, and neither the mesh nor the integrals are recomputed. For a fixed system $F\bm{y}=\bm{b}$, the iteration is
\begin{equation}
\bm{r}^{(k)}=\bm{b}-F\bm{y}^{(k)},\qquad F\delta \bm{y}^{(k)}\simeq \bm{r}^{(k)},\qquad
\bm{y}^{(k+1)}=\bm{y}^{(k)}+\delta \bm{y}^{(k)}.
\label{eq:refine}
\end{equation}
The stored double-precision coefficients and the initial approximation are converted exactly to rational numbers, and residuals and solution updates are evaluated in rational arithmetic. The correction equation is solved in double precision: the residual is rounded to binary64, the correction is computed with a SciPy SuperLU factorisation and converted exactly to rational form, and the approximate solution is then updated. Higher precision thus enters the residual evaluation and the updates but not the matrix assembly, and convergence depends on the conditioning of the system and on this precision configuration \citep{carson2018refinement}.

For $R$, refinement starts from the EIDORS output and is applied to the grounded augmented system. For each Robin implementation, it starts from the computed response matrix $Z$ and is applied to $H_gZ=C_g$; $M_g$ is then formed, the electrode system solved and the potentials recovered in rational arithmetic. Each implementation uses its own SuperLU factorisation. In the 2692-node case, four refinement steps are applied to $R$, $A$, $B_s$ and $C_s$, whereas the platform-native results $P_s$ are not refined. In the 2535-node case, four steps are applied to $R$, $A$ and $P_s$, and the residuals~\eqref{eq:full_residual} and reference errors~\eqref{eq:reference_errors} are recorded at every step.

The results at $k=0$ are the directly computed potentials of each implementation, and those for $k\geq1$ are produced by this common external procedure. The four-step rational approximations $\bm{x}_s^\star$ are distinguished from their double-precision roundings $\operatorname{RN}_{64}(\bm{x}_s^\star)$. Identical roundings after steps 3 and 4 indicate rounding stability but do not certify that they equal the correctly rounded $\widetilde{\bm{x}}_s$.

In the componentwise interpretation of equation~\eqref{eq:exact_error_split}, $\bm{x}_s-\bm{x}_s^\star$ estimates the solution-error contribution and $\bm{x}_s^\star-\bm{x}_{\mathbb Q}$ the assembly-induced solution shift. The two estimates sum exactly to $\bm{x}_s-\bm{x}_{\mathbb Q}$, and the refinement error remaining after four steps enters them with equal magnitude and opposite signs.

\subsection{Parameter comparisons and analytical checks}\label{sec:additional}
Parameter comparisons use the fixed 2692-node mesh with contact-impedance ratios $z/z_0\in\{0.01,0.1,1,10,100\}$ ($z_0=0.125$), inclusion conductivities $\sigma_{\mathrm{inc}}\in\{0.1,10\}$, and non-uniform contact impedances $z_\ell=z_0[1+0.5\sin(2\pi(\ell-1)/16)]$, giving eight conditions including the baseline. Under each condition, $R$ is compared with $A$ computed from the EIDORS matrices of that condition; at the two extreme impedance ratios, NGSolve's $B_s$ and $C_s$ are also included.

The infinity-norm condition number of $H_g$ is estimated as $\|H_g\|_\infty$ times a one-norm estimate of $H_g^{-\mathsf T}$, obtained with the block one-norm estimator \citep{higham2000norm} (block width 2, seed 0); the result is a lower-bound estimate. The condition number $\kappa_\infty(M_g)$ is computed directly from the floating-point matrix formed in each calculation.

The analytical check uses the rectangle $(0,a)\times(0,b)$, with the entire left and right sides serving as electrodes and the top and bottom sides insulated. If the ground node has a horizontal coordinate $x_g$ and the total currents on the left and right electrodes are $+I$ and $-I$, respectively, then
\begin{equation}
\begin{aligned}
u(x)&=\frac{I(x_g-x)}{\sigma b},\\
U_1&=\frac{Ix_g}{\sigma b}+\frac{z_1I}{b},\qquad
U_2=\frac{I(x_g-a)}{\sigma b}-\frac{z_2I}{b}.
\end{aligned}\label{eq:analytic}
\end{equation}
We use $a=2$, $b=1$, $\sigma=1.5$, $z_1=0.125$, $z_2=0.25$ and $I=1$ on a mesh of 315 nodes and 560 P1 triangles. Because this solution is affine in $x$, it belongs to the P1 space, and the exact discrete solution coincides with equation~\eqref{eq:analytic} at the nodes and electrodes. The analytical values are evaluated in rational arithmetic at the stored nodal coordinates and rounded to double precision, and the directly computed results of $R$ and $A$ are compared with them to check the current signs, contact-impedance terms and potential reference.

\section{Results}\label{sec:results}
\subsection{Results from platform-native assembly and solution}\label{sec:primary_results}
All five platforms assembled the 2692-node model through their native Robin boundary interfaces and solved it with the Robin--transadmittance procedure. The maximum absolute electrode-potential differences from the EIDORS output $R$ ranged from $1.732\times10^{-14}$ to $4.641\times10^{-13}$, and the largest normalised backward error was $3.508\times10^{-16}$ (table~\ref{tab:primary}).

\begin{table}[!htbp]\centering\small
\caption{Platform-native Robin--transadmittance implementations compared with the EIDORS conventional augmented CEM for the 2692-node case (directly computed results). Potential differences are taken over all 16 stimulation columns, and backward errors are evaluated with each platform's own assembled equations.}\label{tab:primary}
\begin{tabular*}{\textwidth}{@{\extracolsep{\fill}}llll@{}}\toprule
Platform & $\|\bm{U}_{P_s}-\bm{U}_R\|_{\max}$ & $e_U$ & $\eta_{\infty,P_s}$\\\midrule
DOLFINx & $3.531\times10^{-14}$ & $2.539\times10^{-14}$ & $1.301\times10^{-16}$\\
FreeFEM & $1.910\times10^{-14}$ & $9.165\times10^{-15}$ & $7.402\times10^{-17}$\\
GetFEM & $2.842\times10^{-14}$ & $9.164\times10^{-15}$ & $9.162\times10^{-17}$\\
MFEM & $4.641\times10^{-13}$ & $3.069\times10^{-13}$ & $3.508\times10^{-16}$\\
NGSolve & $1.732\times10^{-14}$ & $7.539\times10^{-15}$ & $1.016\times10^{-16}$\\
\bottomrule
\end{tabular*}
\end{table}

These results show that the CEM can be implemented through the existing Robin boundary interfaces of the tested platforms. Supplementary table~S15 extends the comparison to all nodal potentials, and table~S16 compares $P_s$ with the local-integral-recording implementation $C_s$. The differences in table~\ref{tab:primary} combine the effects of assembly and solution; the shared-matrix comparisons below remove the former by using identical assembled matrices.

\subsection{Consistency of solution formulations with shared assembled matrices}\label{sec:shared_results}
With identical matrix entries, ground node and current inputs, the MATLAB Robin implementation $A$ and the EIDORS augmented implementation $R$ differed by at most $1.243\times10^{-14}$ in the electrode potentials of the 2692-node case. Because the coefficients are shared, this difference arises entirely from the floating-point operations of the two solution procedures, including the treatment of multiple right-hand sides, factorisation, elimination and potential recovery.

After four refinement steps applied separately to $A$ and $R$, the two results differed by at most $2.211\times10^{-69}$, and their double-precision roundings agreed bitwise for all 43328 potential values. The rounded results of the five $B_s$ implementations, which also used the benchmark matrices, were identical to these values, and all roundings were unchanged between the last two steps. Supplementary section~8 gives the corresponding response-equation residuals.

The $C_s$ implementations, whose assembled matrices differ from the benchmark matrices, still differed from $B_s$ after refinement, with maximum absolute electrode-potential differences from $7.011\times10^{-15}$ to $1.508\times10^{-14}$ across platforms (supplementary table~S14). Supplementary sections~3, 4 and~7 analyse these differences componentwise and separate the contributions of local integration and matrix accumulation.

These shared-matrix results agree with the algebraic equivalence of section~\ref{sec:equivalence}. Shared matrices are not, however, exactly assembled matrices; the exact discrete reference below evaluates all implementations, including EIDORS, against the same discrete problem.

\subsection{Accuracy relative to the exact discrete reference}\label{sec:rational_results}
The two independent exact assemblies of the 2535-node case gave identical matrices, and the 40816 exact reference values satisfied the full augmented system with zero residual. Table~\ref{tab:rational} lists the maximum absolute errors of all implementations relative to $\bm{x}_{\mathbb Q}$. For the directly computed results, the maximum electrode-potential errors ranged from $3.371\times10^{-14}$ to $1.087\times10^{-13}$. The errors in the voltages relative to electrode 16, $E_{\Delta U,s}^{(0)}$, were of similar magnitude, from $2.401\times10^{-14}$ to $1.708\times10^{-13}$, whereas the columnwise mean electrode-potential error, removed when both outputs are expressed in the zero-mean reference of section~\ref{sec:grounded_form}, did not exceed $4.579\times10^{-14}$ in magnitude (supplementary table~S28). The errors therefore cannot be attributed entirely to a common potential shift.

\begin{table}[!htbp]\centering\small
\caption{Maximum absolute errors relative to the exact discrete reference for the 2535-node case. Here, $k=0$ denotes directly computed results, and $k=4$ denotes rational approximations after four refinement steps with each implementation's assembled matrices fixed. Errors span all 16 stimulation columns and the relevant nodal or electrode degrees of freedom.}\label{tab:rational}
\begin{tabular*}{\textwidth}{@{\extracolsep{\fill}}llll@{}}\toprule
Implementation & $E_{u,s}^{(0)}$ & $E_{U,s}^{(0)}$ & $E_{U,s}^{(4)}$\\\midrule
EIDORS augmented ($R$) & $7.875\times10^{-14}$ & $7.809\times10^{-14}$ & $1.086\times10^{-13}$\\
MATLAB Robin ($A$) & $1.094\times10^{-13}$ & $1.087\times10^{-13}$ & $1.086\times10^{-13}$\\
DOLFINx ($P_s$) & $3.380\times10^{-14}$ & $3.371\times10^{-14}$ & $3.092\times10^{-14}$\\
FreeFEM ($P_s$) & $5.456\times10^{-14}$ & $5.402\times10^{-14}$ & $5.585\times10^{-14}$\\
GetFEM ($P_s$) & $4.523\times10^{-14}$ & $4.558\times10^{-14}$ & $8.072\times10^{-14}$\\
MFEM ($P_s$) & $1.045\times10^{-13}$ & $1.014\times10^{-13}$ & $2.994\times10^{-14}$\\
NGSolve ($P_s$) & $4.893\times10^{-14}$ & $5.323\times10^{-14}$ & $4.427\times10^{-14}$\\
\bottomrule
\end{tabular*}
\end{table}

The stored finite element matrices of every implementation differed from the exactly assembled matrices, by up to $2.691\times10^{-14}$ in $H_g$ (supplementary table~S20). Reducing the solution error for an assembled system therefore need not bring the approximation closer to $\bm{x}_{\mathbb Q}$: after four refinement steps, the maximum electrode-potential errors had increased for EIDORS, FreeFEM and GetFEM and decreased for the other implementations (table~\ref{tab:rational}).

\subsection{Residuals and reference errors during refinement}\label{sec:refinement_results}
Figure~\ref{fig:refinement} compares the grounded augmented-system residuals with the electrode-potential errors relative to $\bm{x}_{\mathbb Q}$. After four refinement steps, the residuals had fallen to between $3.805\times10^{-73}$ and $3.224\times10^{-70}$, whereas the reference errors remained between $2.994\times10^{-14}$ and $1.086\times10^{-13}$; from the first step onwards, the reference errors were unchanged to four significant figures. The residual measures how well the stored equations are satisfied, whereas the reference error measures the distance to the exact discrete solution. Voltage-difference errors also increased for some implementations and decreased for others, not always in the same direction as the electrode-potential errors (supplementary table~S28).

\begin{figure}[!htbp]\centering
\includegraphics[width=\linewidth]{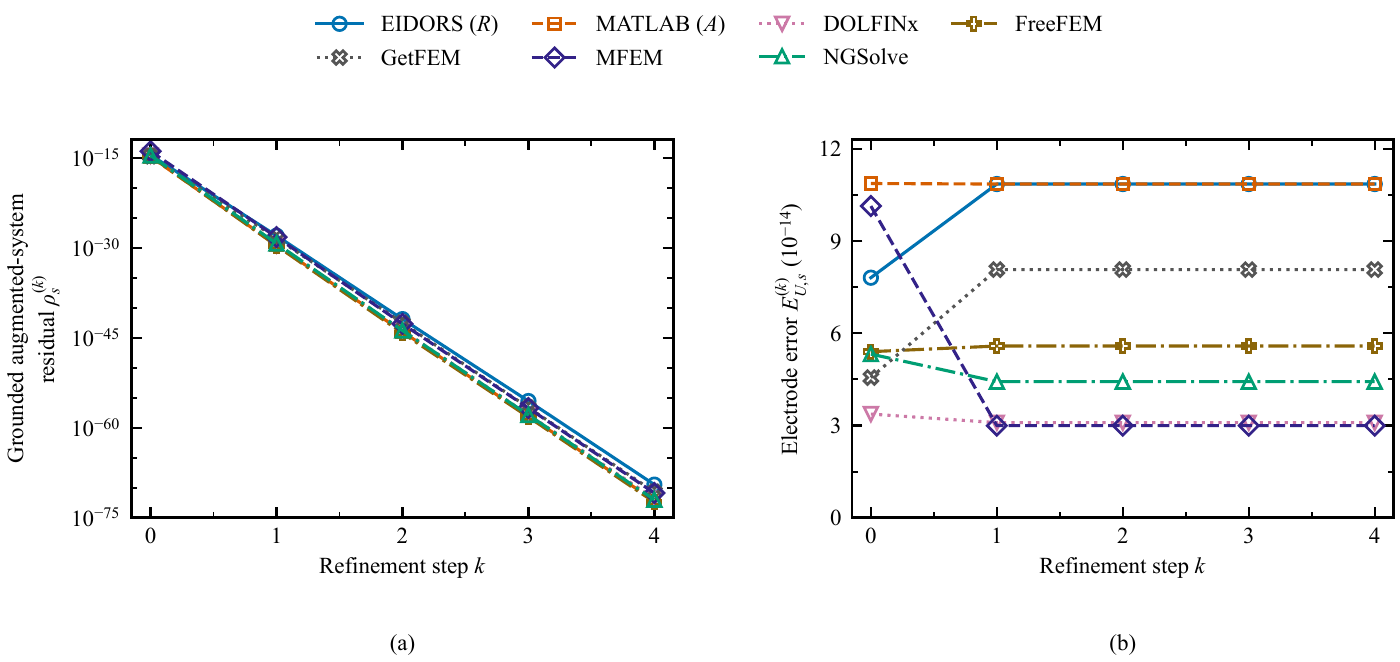}
\caption{Iterative refinement with fixed assembled matrices in the 2535-node case. (a) Grounded augmented-system residual $\rho_s^{(k)}$, evaluated with each implementation's own assembled matrix and excluding the ground row. (b) Maximum absolute electrode-potential error $E_{U,s}^{(k)}$ relative to the exact discrete reference, in units of $10^{-14}$. Both panels evaluate the same potential approximation at each step; $k=0$ denotes the directly computed results, and later steps follow section~\ref{sec:refinement_method}. The curves for $R$ and $A$ in (b) overlap from $k=1$ onwards.}\label{fig:refinement}
\end{figure}

The maximum electrode-potential error increased from $7.809\times10^{-14}$ to $1.086\times10^{-13}$ for EIDORS but decreased from $1.014\times10^{-13}$ to $2.994\times10^{-14}$ for MFEM. According to equation~\eqref{eq:exact_error_split}, the outcome depends on how the solution error and the assembly-induced solution shift combine. At electrode 8 in stimulation column 9, where the directly computed EIDORS result had its largest error, the estimated contributions are
\begin{equation}
\underbrace{-7.809\times10^{-14}}_{\substack{\text{total error of the}\\\text{directly computed result}}}
\approx
\underbrace{+3.048\times10^{-14}}_{\substack{\text{estimated solution-error}\\\text{contribution}}}
+\underbrace{(-1.086\times10^{-13})}_{\substack{\text{estimated assembly}\\\text{contribution}}}.
\label{eq:eidors_cancellation}
\end{equation}
The positive estimated solution-error contribution partly cancels the negative estimated assembly contribution; refinement largely removes the former, and the error at this component increases. For NGSolve, the corresponding component, at electrode 9 in stimulation column 10, gives
\begin{equation}
\underbrace{+5.323\times10^{-14}}_{\substack{\text{total error of the}\\\text{directly computed result}}}
\approx
\underbrace{+8.962\times10^{-15}}_{\substack{\text{estimated solution-error}\\\text{contribution}}}
+\underbrace{(+4.427\times10^{-14})}_{\substack{\text{estimated assembly}\\\text{contribution}}}.
\label{eq:ngsolve_alignment}
\end{equation}
Here the two estimated contributions have the same sign, and refinement reduces the error at this component. Supplementary table~S25 lists these signed estimates for all implementations.

In the same case, $R$ and $A$, which share the EIDORS-assembled matrices, differed by at most $3.616\times10^{-69}$ after four refinement steps, and their 40816 values rounded to binary64 agreed bitwise; yet both retained a maximum electrode-potential error of about $1.086\times10^{-13}$. Agreement between solution formulations and accuracy relative to the exact discrete reference are thus separate properties. The very small residuals after refinement refer to the rational approximations rather than to their double-precision roundings.

\subsection{Parameter comparisons and analytical checks}\label{sec:additional_results}
Across the eight parameter conditions, the maximum absolute electrode-potential difference between $R$ and $A$ did not exceed $1.263\times10^{-13}$, and the largest normalised backward error of the Robin solution was $1.542\times10^{-16}$ (supplementary tables~S10--S12). The largest difference occurred at the smallest contact impedance, $z/z_0=0.01$. Because the contact impedance affects the potential magnitudes, the response calculations and the electrode system simultaneously, the condition number of $H_g$ alone does not account for these differences.

For the rectangular domain, whose analytical voltage is $U_1-U_2=41/24$, the maximum absolute nodal-potential differences from the rounded analytical reference were $6.661\times10^{-14}$ for $R$ and $7.994\times10^{-15}$ for $A$, and the electrode-potential differences were $6.550\times10^{-14}$ and $7.550\times10^{-15}$, respectively. These results support the correct implementation of the current signs, contact-impedance terms and potential reference in this test case.

\section{Discussion}\label{sec:discussion}
The Robin--transadmittance method separates total electrode-current constraints from nodal finite element solves. A general-purpose finite element package solves Robin subproblems with prescribed electrode potentials, and the transadmittance system imposes the total electrode-current constraints. Electrode potentials remain unknowns of the full CEM problem but need not be included in the package's nodal system. The method therefore preserves the CEM contact-impedance condition without adding electrode degrees of freedom to the finite element subproblem. The implementations on five platforms show that the approach does not depend on a particular software package.

The CEM is easily implemented in a custom finite element code by adding one unknown potential per electrode to the nodal system, as in EIDORS \citep{vauhkonen1999three,polydorides2002matlab,adler2006eidors}. There are nevertheless good reasons to use general-purpose finite element software. Such packages are typically well maintained and widely tested and, depending on the package, offer flexible polynomial-order refinement, adaptive refinement near electrode edges, coupling with boundary elements and, where needed, infinite elements or perfectly matched layers. Solver libraries such as PETSc and hypre provide parallel solvers and preconditioners \citep{balay1997petsc,falgout2002hypre}, and some packages run on GPUs \citep{anderson2021mfem}. Python interfaces also link forward models to reconstruction and machine-learning workflows; PyEIDORS, for example, combines FEniCS-based CEM forward modelling with PyTorch-enabled inverse solvers \citep{chen2026pyeidors}.

The result in theorem~\ref{thm:equivalence} establishes the equivalence before choosing a potential reference: the same blocks and an invertible Robin matrix suffice. In turn, proposition~\ref{prop:grounding} connects this result to the nodal implementation used for the EIDORS comparisons, and corollary~\ref{cor:gauge} explains why exact voltage differences are independent of the reference. When the blocks are identical, the augmented and Robin--transadmittance formulations differ only in their solution procedures, and this difference is what the shared-matrix comparisons measure. Close or even bitwise agreement, however, concerns only the implementations being compared. With common current inputs and grounding, implementations that share assembled coefficients have the same assembly-induced solution shift, although their finite-precision solution errors may differ. The $R$--$A$ comparison of section~\ref{sec:refinement_results} illustrates this distinction: their directly computed errors differed, whereas their refined results, which agreed bitwise after rounding to binary64, retained the same maximum electrode-potential error of about $1.086\times10^{-13}$. Inter-software comparisons therefore need to take into account where the assembled matrices come from.

The rational-arithmetic case clarifies what refinement with fixed assembled matrices can achieve. Refinement drives the approximation towards the solution of the stored equations, not towards that of the exactly assembled problem. Where the solution error and the assembly-induced shift have opposite signs, the directly computed result can be closer to the exact discrete reference than a highly refined one, as for EIDORS, FreeFEM and GetFEM in table~\ref{tab:rational}. Refinement served in this study as a diagnostic; it is not part of the Robin--transadmittance method and is not needed in routine use.

Finite-precision comparisons also depend on how the potential reference is imposed. Let $\bm{w}$ contain all nodal and electrode potentials, and let $\bm{q}=[\bm1_N^{\mathsf T},\bm1_L^{\mathsf T}]^{\mathsf T}$. Exact assembly gives $F_0\bm{q}=0$, but a stored matrix may have $\widehat F_0\bm{q}\ne0$, in which case $\widehat F_0(\bm{w}-c\bm{q})=\widehat F_0\bm{w}-c\widehat F_0\bm{q}$. Subtracting the mean electrode potential from a grounded output therefore need not reproduce bitwise the solution of a separately solved system with a zero-mean constraint (section~\ref{sec:grounded_form}). Using the same ground node in all implementations prevents such differences in the potential reference from being confounded with differences between formulations.

For fixed geometry, conductivity and contact impedances, the response matrix and the factorisations can be reused for new current inputs. The implementation presented here first computes the full set of electrode responses, irrespective of how many current patterns are needed. With $L\ll N$, the main setup cost is generally in constructing these responses, including the large Robin solves, rather than in solving the small electrode system. Because the augmented formulation can likewise reuse its factorisation, the stepwise procedure does not by itself confer a speed advantage. When only a few current patterns are required, as in some applications other than full EIT or in three-dimensional problems where the applied currents need not span the space of possible currents, computing all responses can be more expensive than solving the augmented system for those patterns. A dedicated CEM implementation may therefore be preferable for a fast inverse solver. The practical benefit of the Robin--transadmittance procedure is that no CEM-specific functionality is required from the finite element package. Portability does, however, require consistent finite element spaces, material and electrode representations, boundary integration and current sign conventions, and a common mesh and numbering do not by themselves give bitwise-identical results when integration, accumulation or solution sequences differ.

\Needspace{3\baselineskip}
The numerical validation uses two-dimensional, conforming P1 elements to enable controlled comparisons with the first-order EIDORS discretisation and an exact rational-arithmetic reference. The block-elimination argument is not specific to this choice: it applies to other conforming finite element discretisations, including higher-order and three-dimensional models, provided that the same CEM blocks are used and the Robin block is invertible. The explicit P1 energy construction retains its stated scope, and the numerical results concern only the discretisations and software configurations tested. Each exact rational reference quantifies numerical error within its own discrete problem, and the rectangular case checks an affine solution that the P1 space represents exactly; neither addresses convergence to the continuous CEM solution. The results therefore support the discrete equivalence and portability of the formulation rather than a general ranking of solver accuracy or computational efficiency.

\section{Conclusions}\label{sec:conclusions}
We have presented a portable implementation of the current-driven CEM through standard Robin boundary interfaces in general-purpose finite element software. The Robin--transadmittance and conventional augmented CEM procedures have the same solution set for identical blocks and an invertible Robin matrix. This equivalence is independent of the potential reference. The discrete energy argument, equivalence theorem and supporting reference results, including the zero-mean/pseudoinverse construction, are formalised in Lean 4; the numerical comparisons adopt the same nodal reference as EIDORS. All five tested platforms assembled the CEM finite element blocks through their native Robin boundary interfaces, and the solutions were obtained by combining nodal response solves with an electrode transadmittance system. In the 2692-node case, maximum absolute electrode-potential differences from EIDORS ranged from $1.732\times10^{-14}$ to $4.641\times10^{-13}$. In the 2535-node case, the maximum absolute errors of the directly computed electrode potentials relative to the exact discrete reference ranged from $3.371\times10^{-14}$ to $1.087\times10^{-13}$.

The exact reference also shows that smaller assembled-system residuals after refinement need not imply smaller errors relative to the exact discrete solution.

\ifdefempty{\PaperAcknowledgments}{}{\ack{\PaperAcknowledgments}}
\ifdefempty{\PaperAIStatement}{}{\section*{AI use disclosure}\noindent\PaperAIStatement\par}
\ifdefempty{\PaperFunding}{}{\funding{\PaperFunding}}
\ifdefempty{\PaperAuthorContributions}{}{\roles{\PaperAuthorContributions}}
\ifdefempty{\PaperCompetingInterests}{}{\section*{Competing interests}\PaperCompetingInterests}

\data{
The reproduction package supporting this study is openly available on Zenodo at \url{https://doi.org/10.5281/zenodo.23239050}. It contains the fixed discrete inputs, the EIDORS and five-platform implementations, the programs for exact rational assembly and solution, iterative refinement and error analysis, the scripts for exporting numerical summaries and drawing the figures, and the Lean 4 proof sources with pinned dependencies, together with instructions for use. Saved summaries support correspondence checks; full precomputed solution arrays are not included. The reported quantities can be regenerated with the documented software versions.
\ifdefempty{\PaperDataAvailability}{}{\par\PaperDataAvailability}
}

\suppdata{Supplementary material (PDF): additional numerical comparisons, assembly and refinement diagnostics, exact discrete-reference results, voltage-difference errors and the correspondence between paper results and Lean 4 proofs.}

\begingroup
\interlinepenalty=10000
\setlength{\bibsep}{3pt}
\bibliographystyle{abbrvnat}
\bibliography{references}
\endgroup
\end{document}


\articletype{Supplementary material}
\title{Supplementary material}
\PaperAuthorBlock

\noindent\textbf{Related article:} Implementing the EIT complete electrode model through Robin boundary conditions.\par
\noindent\textbf{Description:} Additional numerical comparisons, assembly and refinement diagnostics, exact discrete-reference results, voltage-difference errors, and correspondence between paper results and Lean 4 proofs.\par

\section{Computational implementations and comparisons}
This supplement gives further results for the test cases of the main text. All numerical quantities are dimensionless under the convention stated in main-text section~3.2, including coordinates, conductivities, contact impedances, currents, potentials, residuals and absolute errors. The same convention applies to every table and figure in this supplement. Sections 1--4, 7 and 8 concern the 2692-node comparisons with EIDORS, section 9 the 2535-node exact discrete benchmark, and section 5 the parameter variations and the rectangular analytical case. Section 6 describes the software configurations, and section 10 relates the algebraic propositions of the main text to the Lean 4 theorems.

The implementation labels follow table 2 of the main text. $R$ is the conventional augmented complete electrode model (CEM) result of the EIDORS forward solver; $P_s$ is the platform-native Robin--transadmittance result on platform $s$; $A$ is our MATLAB Robin--transadmittance implementation; $B_s$ applies the solver of platform $s$ to the EIDORS matrices $H_g,C_g,D$, which $A$ also uses; and $C_s$ uses platform-computed element and electrode-edge integrals whose accumulation is recorded. The 2692-node case includes all five groups, whereas the 2535-node case includes $R$, $A$ and $P_s$. The parameter conditions compare $R$ with $A$, with NGSolve's $B_s$ and $C_s$ added at the two extreme impedance ratios, and the rectangular case compares $R$ with $A$ for a single current pattern.

Unless marked with $\star$ or a refinement-step index, tabulated results are directly computed, that is, obtained from each implementation's assembly and double-precision solution. Refinement steps $k=1,\ldots,4$ follow main-text section 3.5, with $\bm{x}_s^\star=\bm{x}_s^{(4)}$ and its double-precision rounding $\operatorname{RN}_{64}(\bm{x}_s^\star)$.

For the 2692-node comparisons, including the parameter conditions of section 5, the relative nodal- and electrode-potential differences are $e_u=\|\bm{u}_s-\bm{u}_R\|_F/\|\bm{u}_R\|_F$ and $e_U=\|\bm{U}_s-\bm{U}_R\|_F/\|\bm{U}_R\|_F$, normalised by the directly computed EIDORS output. Maximum absolute differences are taken over all 16 stimulation columns and the corresponding degrees of freedom, with the ground row restored to zero. Backward errors and electrode-equation residuals are evaluated with the stored grounded matrix $\widehat F_{g,s}$ of each implementation.

\begin{table}[H]\centering\small
\caption{Differences of each implementation from the EIDORS benchmark $R$ in the 2692-node case (directly computed results). Table~\ref{tab:s15} gives the corresponding platform-native results.}\label{tab:s1}
\begin{tabular*}{\textwidth}{@{\extracolsep{\fill}}lllll@{}}
\toprule
Implementation & \makecell[l]{Maximum absolute\\nodal difference} & $e_u$ & \makecell[l]{Maximum absolute\\electrode difference} & $e_U$\\\midrule
$R$ & $0$ & $0$ & $0$ & $0$\\
$A$ & $1.221\times10^{-14}$ & $1.012\times10^{-14}$ & $1.243\times10^{-14}$ & $5.246\times10^{-15}$\\
DOLFINx ($B_s$) & $3.619\times10^{-14}$ & $5.356\times10^{-14}$ & $3.642\times10^{-14}$ & $2.622\times10^{-14}$\\
DOLFINx ($C_s$) & $5.018\times10^{-14}$ & $6.215\times10^{-14}$ & $4.841\times10^{-14}$ & $3.063\times10^{-14}$\\
FreeFEM ($B_s$) & $8.438\times10^{-15}$ & $6.779\times10^{-15}$ & $8.438\times10^{-15}$ & $3.526\times10^{-15}$\\
FreeFEM ($C_s$) & $1.599\times10^{-14}$ & $1.646\times10^{-14}$ & $1.421\times10^{-14}$ & $7.946\times10^{-15}$\\
GetFEM ($B_s$) & $1.110\times10^{-14}$ & $7.783\times10^{-15}$ & $1.044\times10^{-14}$ & $4.057\times10^{-15}$\\
GetFEM ($C_s$) & $3.397\times10^{-14}$ & $2.064\times10^{-14}$ & $2.931\times10^{-14}$ & $1.012\times10^{-14}$\\
MFEM ($B_s$) & $4.647\times10^{-13}$ & $6.300\times10^{-13}$ & $4.605\times10^{-13}$ & $3.070\times10^{-13}$\\
MFEM ($C_s$) & $4.676\times10^{-13}$ & $6.309\times10^{-13}$ & $4.641\times10^{-13}$ & $3.080\times10^{-13}$\\
NGSolve ($B_s$) & $3.508\times10^{-14}$ & $5.179\times10^{-14}$ & $3.553\times10^{-14}$ & $2.533\times10^{-14}$\\
NGSolve ($C_s$) & $2.709\times10^{-14}$ & $2.982\times10^{-14}$ & $2.620\times10^{-14}$ & $1.490\times10^{-14}$\\
\bottomrule\end{tabular*}\end{table}

\section{Equation residuals and matrix condition numbers}
The residuals in tables~\ref{tab:s2} and \ref{tab:s3} are obtained by substituting $\bm{x}_{g,s}=[\bm{v}_s^{\mathsf T},\bm{U}_s^{\mathsf T}]^{\mathsf T}$ into the grounded nodal--electrode system, with the EIDORS matrices for $R$, $A$ and $B_s$ and the stored matrices of each $C_s$ otherwise. Table~\ref{tab:s3} reports, over all 16 stimulation columns,
\[\rho_s=\left\|\begin{bmatrix}0\\\bm{I}\end{bmatrix}-\widehat F_{g,s}\begin{bmatrix}\bm{v}_s\\\bm{U}_s\end{bmatrix}\right\|_{\max}\]
both for the refined rational approximations and for their double-precision roundings. Figure~\ref{fig:shared_response} shows instead the residual of the response equations $H_gZ=C_g$, and the 2535-node residuals of main-text figure 4(a) are listed in table~\ref{tab:s24}.
\begin{table}[H]\centering\small
\caption{Normalised backward errors, electrode-equation residuals and matrix condition numbers for the directly computed results. Residuals are maxima over all 16 stimulation columns. $\widehat\kappa_\infty(H_g)$ is an estimate, and $\kappa_\infty(M_g)$ is computed from the dense matrix formed in each calculation; for $R$, $M_g$ is formed from the EIDORS matrices, although the EIDORS solver itself does not use it.}\label{tab:s2}
\begin{tabular*}{\textwidth}{@{\extracolsep{\fill}}lllll@{}}
\toprule
Implementation & $\eta_\infty$ & \makecell[l]{Electrode-equation\\residual} & $\widehat\kappa_\infty(H_g)$ & $\kappa_\infty(M_g)$\\\midrule
$R$ & $8.923\times10^{-17}$ & $1.771\times10^{-15}$ & $3.754\times10^{3}$ & $1.484\times10^{1}$\\
$A$ & $7.364\times10^{-17}$ & $5.699\times10^{-16}$ & $3.754\times10^{3}$ & $1.484\times10^{1}$\\
DOLFINx ($B_s$) & $1.536\times10^{-16}$ & $6.478\times10^{-16}$ & $3.754\times10^{3}$ & $1.484\times10^{1}$\\
DOLFINx ($C_s$) & $1.265\times10^{-16}$ & $7.141\times10^{-16}$ & $3.754\times10^{3}$ & $1.484\times10^{1}$\\
FreeFEM ($B_s$) & $8.610\times10^{-17}$ & $6.091\times10^{-16}$ & $3.754\times10^{3}$ & $1.484\times10^{1}$\\
FreeFEM ($C_s$) & $7.076\times10^{-17}$ & $6.242\times10^{-16}$ & $3.754\times10^{3}$ & $1.484\times10^{1}$\\
GetFEM ($B_s$) & $1.060\times10^{-16}$ & $7.816\times10^{-16}$ & $3.754\times10^{3}$ & $1.484\times10^{1}$\\
GetFEM ($C_s$) & $1.143\times10^{-16}$ & $7.485\times10^{-16}$ & $3.754\times10^{3}$ & $1.484\times10^{1}$\\
MFEM ($B_s$) & $4.258\times10^{-16}$ & $6.322\times10^{-16}$ & $3.754\times10^{3}$ & $1.484\times10^{1}$\\
MFEM ($C_s$) & $3.366\times10^{-16}$ & $4.581\times10^{-16}$ & $3.754\times10^{3}$ & $1.484\times10^{1}$\\
NGSolve ($B_s$) & $1.173\times10^{-16}$ & $6.042\times10^{-16}$ & $3.754\times10^{3}$ & $1.484\times10^{1}$\\
NGSolve ($C_s$) & $7.618\times10^{-17}$ & $4.746\times10^{-16}$ & $3.754\times10^{3}$ & $1.484\times10^{1}$\\
\bottomrule\end{tabular*}\end{table}

\begin{table}[H]\centering\small
\caption{Grounded augmented-system residuals $\rho_s$ of the refined rational approximations and of their double-precision roundings, evaluated with the same matrices as in table~\ref{tab:s2}.}\label{tab:s3}
\begin{tabular*}{\textwidth}{@{\extracolsep{\fill}}lll@{}}
\toprule
Implementation & $\rho_s(\bm{x}_{g,s}^\star)$ & $\rho_s(\operatorname{RN}_{64}(\bm{x}_{g,s}^\star))$\\\midrule
$R$ & $1.724\times10^{-70}$ & $5.879\times10^{-16}$\\
$A$ & $1.963\times10^{-73}$ & $5.879\times10^{-16}$\\
DOLFINx ($B_s$) & $1.187\times10^{-72}$ & $5.879\times10^{-16}$\\
DOLFINx ($C_s$) & $8.413\times10^{-73}$ & $5.972\times10^{-16}$\\
FreeFEM ($B_s$) & $1.799\times10^{-73}$ & $5.879\times10^{-16}$\\
FreeFEM ($C_s$) & $1.550\times10^{-73}$ & $5.680\times10^{-16}$\\
GetFEM ($B_s$) & $1.945\times10^{-73}$ & $5.879\times10^{-16}$\\
GetFEM ($C_s$) & $6.371\times10^{-73}$ & $5.393\times10^{-16}$\\
MFEM ($B_s$) & $1.208\times10^{-71}$ & $5.879\times10^{-16}$\\
MFEM ($C_s$) & $8.793\times10^{-72}$ & $6.401\times10^{-16}$\\
NGSolve ($B_s$) & $1.064\times10^{-72}$ & $5.879\times10^{-16}$\\
NGSolve ($C_s$) & $6.548\times10^{-73}$ & $5.770\times10^{-16}$\\
\bottomrule\end{tabular*}\end{table}

\section{Componentwise analysis of differences between implementations}
The difference between a local-integral-recording implementation and the EIDORS benchmark can be written componentwise as
\[
\bm{x}_{C_s}-\bm{x}_R=(\bm{x}_{C_s}-\bm{x}_{B_s})+(\bm{x}_{B_s}-\bm{x}_A)+(\bm{x}_A-\bm{x}_R),
\]
where the three terms reflect the assembly source, the platform-specific solution procedure and the solution formulation, respectively. The first term contains both the effect of the changed assembly and the accompanying change in solution error. Refined approximations are used below to estimate these contributions.

For each platform, table~\ref{tab:s4} evaluates the terms at the electrode-potential component $(\ell,j)$, with electrode $\ell$ and stimulation column $j$ numbered from 1, at which $C_s$ differs most from $R$:
\[\bm{U}_{C_s}-\bm{U}_R=(\bm{U}_{C_s}^\star-\bm{U}_{B_s}^\star)+(\bm{U}_{B_s}^\star-\bm{U}_R^\star)+\bm{e}_{C_s}-\bm{e}_R,\]
where $\bm{e}_s=\bm{U}_s-\bm{U}_s^\star$ is the difference between the directly computed and refined results of implementation $s$. The first two terms, labelled `estimated assembly effect' and `shared-matrix solution difference' in table~\ref{tab:s4}, estimate the change in solution caused by the different assembled matrices and the difference remaining after refinement with the shared benchmark matrices. Tables~\ref{tab:s1} and \ref{tab:s14} give the corresponding maxima over all components. Because the tabulated values are rounded, their displayed sums may differ in the last digit.
\begin{table}[H]\centering\footnotesize
\caption{Componentwise decomposition of the difference between each local-integral-recording implementation $C_s$ and the EIDORS benchmark $R$ in the 2692-node case. All terms in a row refer to the same electrode-potential component, and signs give the direction of the deviation from $R$.}\label{tab:s4}
\footnotesize\setlength{\tabcolsep}{3pt}
\begin{tabular*}{\textwidth}{@{\extracolsep{\fill}}lllllll@{}}
\toprule
Platform & $(\ell,j)$ & \makecell[l]{Total\\difference} & \makecell[l]{Estimated\\assembly effect} & \makecell[l]{Shared-matrix\\solution difference} & $e_{C_s}$ & $-e_R$\\\midrule
DOLFINx & (7,8) & $-4.841\times10^{-14}$ & $-5.180\times10^{-15}$ & $-1.187\times10^{-69}$ & $-3.778\times10^{-14}$ & $-5.443\times10^{-15}$\\
FreeFEM & (11,12) & $+1.421\times10^{-14}$ & $+5.693\times10^{-15}$ & $+9.155\times10^{-70}$ & $+4.362\times10^{-15}$ & $+4.156\times10^{-15}$\\
GetFEM & (8,9) & $+2.931\times10^{-14}$ & $+1.508\times10^{-14}$ & $+1.965\times10^{-70}$ & $+1.318\times10^{-14}$ & $+1.046\times10^{-15}$\\
MFEM & (10,11) & $-4.641\times10^{-13}$ & $+6.371\times10^{-15}$ & $+5.090\times10^{-70}$ & $-4.734\times10^{-13}$ & $+2.967\times10^{-15}$\\
NGSolve & (5,6) & $-2.620\times10^{-14}$ & $-1.645\times10^{-16}$ & $-9.496\times10^{-70}$ & $-2.164\times10^{-14}$ & $-4.394\times10^{-15}$\\
\bottomrule\end{tabular*}\end{table}

For MFEM, the total difference of $-4.641\times10^{-13}$ at electrode 10 and stimulation column 11 consists mainly of the solution-error term $e_{C_s}=-4.734\times10^{-13}$, whereas the estimated assembly effect is only $+6.371\times10^{-15}$. This is consistent with the similar difference of MFEM's $B_s$, which uses the benchmark matrices: the larger MFEM differences in table~\ref{tab:s1} arise mainly from its solution procedure.

\section{Local integration and accumulation of finite element matrices}
This section separates local integration from subsequent accumulation in the matrices of $C_s$. For each matrix block, let $G_s$ be the assembled matrix and $L_s$ the exact (rational) sum of its stored local floating-point entries according to the mesh connectivity, with the subscript $R$ denoting EIDORS. Each entry then satisfies $G_s-G_R=(L_s-L_R)+[(G_s-L_s)-(G_R-L_R)]$. The first term collects the differences from geometric mapping, numerical integration and local arithmetic up to the element or electrode-edge matrix returned by the integration routine; the second collects those from accumulation into the global matrices and combination of blocks. The matrices are compared before the ground degree of freedom is deleted, $A$ denotes the Robin boundary matrix, and the EIDORS block $H$ is the exported benchmark matrix. Maxima in different columns of tables~\ref{tab:s5}--\ref{tab:s9} may occur at different entries and therefore do not add.

Let $\mathcal{T}$ denote the field matrix returned by the EIDORS interface. Its columns index all nodal--electrode degrees of freedom, and its rows represent domain-element and electrode-boundary contributions. For the two-dimensional setting with piecewise-constant scalar conductivity considered here, each triangle and each electrode edge contributes two rows. The diagonal weight matrix $W$ contains the element conductivity for domain rows and 1 for electrode-boundary rows. Geometric factors are included in $\mathcal{T}$, with electrode-edge lengths and contact impedances incorporated through the boundary factors. Before grounding, the EIDORS finite element matrix is obtained by symmetrising the weighted product $(\mathcal{T}^{\mathsf T}W)\mathcal{T}$, where $\operatorname{sym}(X)=(X+X^{\mathsf T})/2$.

For triangle $e$, the corresponding two rows and three nodal columns define $\mathcal{T}_e$. Its local domain contribution is $\operatorname{sym}[(\mathcal{T}_e^{\mathsf T}W_e)\mathcal{T}_e]$, where $W_e=\sigma_e I_2$ and $I_2$ is the identity matrix of order two. For electrode edge $b$, the corresponding two rows and the three columns for its two nodes and one electrode define $\mathcal{T}_b$. Local contributions to $A,C,D$ are extracted from $\operatorname{sym}(\mathcal{T}_b^{\mathsf T}\mathcal{T}_b)$, with the nodal--electrode coupling entries negated to obtain $C$. Exact accumulation of the stored local binary64 entries using the common connectivity defines $L_R$, with $L_R^H=L_R^K+L_R^A$. The matrices $K$ and $A$ are formed from the global products of the domain and boundary rows, respectively, and the resulting $H$, $C$ and $D$ agree bitwise with the blocks exported by EIDORS, so the local records reproduce the benchmark assembly.
\begin{table}[H]\centering\small
\caption{Decomposition of assembled finite element matrix differences for DOLFINx. The decomposition identity holds exactly for every entry; each column reports its own maximum absolute value.}\label{tab:s5}
\begin{tabular*}{\textwidth}{@{\extracolsep{\fill}}llll@{}}
\toprule
Block & \makecell[l]{Assembled-matrix\\difference} & \makecell[l]{Local-integration\\term} & \makecell[l]{Subsequent-arithmetic\\term}\\\midrule
$K$ & $5.329\times10^{-15}$ & $4.441\times10^{-15}$ & $2.442\times10^{-15}$\\
$A$ & $1.110\times10^{-16}$ & $8.327\times10^{-17}$ & $2.776\times10^{-17}$\\
$C$ & $1.110\times10^{-16}$ & $8.327\times10^{-17}$ & $4.163\times10^{-17}$\\
$D$ & $8.882\times10^{-16}$ & $4.718\times10^{-16}$ & $5.274\times10^{-16}$\\
$H$ & $5.329\times10^{-15}$ & $4.441\times10^{-15}$ & $2.442\times10^{-15}$\\
\bottomrule\end{tabular*}\end{table}

\begin{table}[H]\centering\small
\caption{Decomposition of assembled finite element matrix differences for FreeFEM. The decomposition identity holds exactly for every entry; each column reports its own maximum absolute value.}\label{tab:s6}
\begin{tabular*}{\textwidth}{@{\extracolsep{\fill}}llll@{}}
\toprule
Block & \makecell[l]{Assembled-matrix\\difference} & \makecell[l]{Local-integration\\term} & \makecell[l]{Subsequent-arithmetic\\term}\\\midrule
$K$ & $2.665\times10^{-15}$ & $2.109\times10^{-15}$ & $2.442\times10^{-15}$\\
$A$ & $1.110\times10^{-16}$ & $8.327\times10^{-17}$ & $2.776\times10^{-17}$\\
$C$ & $8.327\times10^{-17}$ & $8.327\times10^{-17}$ & $4.163\times10^{-17}$\\
$D$ & $8.882\times10^{-16}$ & $4.718\times10^{-16}$ & $5.551\times10^{-16}$\\
$H$ & $2.665\times10^{-15}$ & $2.109\times10^{-15}$ & $2.442\times10^{-15}$\\
\bottomrule\end{tabular*}\end{table}

\begin{table}[H]\centering\small
\caption{Decomposition of assembled finite element matrix differences for GetFEM. The decomposition identity holds exactly for every entry; each column reports its own maximum absolute value.}\label{tab:s7}
\begin{tabular*}{\textwidth}{@{\extracolsep{\fill}}llll@{}}
\toprule
Block & \makecell[l]{Assembled-matrix\\difference} & \makecell[l]{Local-integration\\term} & \makecell[l]{Subsequent-arithmetic\\term}\\\midrule
$K$ & $3.286\times10^{-14}$ & $3.109\times10^{-14}$ & $2.442\times10^{-15}$\\
$A$ & $3.886\times10^{-16}$ & $3.747\times10^{-16}$ & $2.776\times10^{-17}$\\
$C$ & $5.274\times10^{-16}$ & $5.274\times10^{-16}$ & $4.163\times10^{-17}$\\
$D$ & $4.219\times10^{-15}$ & $3.691\times10^{-15}$ & $7.216\times10^{-16}$\\
$H$ & $3.286\times10^{-14}$ & $3.109\times10^{-14}$ & $2.442\times10^{-15}$\\
\bottomrule\end{tabular*}\end{table}

\begin{table}[H]\centering\small
\caption{Decomposition of assembled finite element matrix differences for MFEM. The decomposition identity holds exactly for every entry; each column reports its own maximum absolute value.}\label{tab:s8}
\begin{tabular*}{\textwidth}{@{\extracolsep{\fill}}llll@{}}
\toprule
Block & \makecell[l]{Assembled-matrix\\difference} & \makecell[l]{Local-integration\\term} & \makecell[l]{Subsequent-arithmetic\\term}\\\midrule
$K$ & $2.665\times10^{-15}$ & $1.998\times10^{-15}$ & $2.331\times10^{-15}$\\
$A$ & $8.327\times10^{-17}$ & $6.939\times10^{-17}$ & $2.776\times10^{-17}$\\
$C$ & $8.327\times10^{-17}$ & $6.939\times10^{-17}$ & $4.163\times10^{-17}$\\
$D$ & $8.882\times10^{-16}$ & $4.718\times10^{-16}$ & $5.274\times10^{-16}$\\
$H$ & $2.665\times10^{-15}$ & $1.998\times10^{-15}$ & $2.331\times10^{-15}$\\
\bottomrule\end{tabular*}\end{table}

\begin{table}[H]\centering\small
\caption{Decomposition of assembled finite element matrix differences for NGSolve. The decomposition identity holds exactly for every entry; each column reports its own maximum absolute value.}\label{tab:s9}
\begin{tabular*}{\textwidth}{@{\extracolsep{\fill}}llll@{}}
\toprule
Block & \makecell[l]{Assembled-matrix\\difference} & \makecell[l]{Local-integration\\term} & \makecell[l]{Subsequent-arithmetic\\term}\\\midrule
$K$ & $2.665\times10^{-15}$ & $2.220\times10^{-15}$ & $2.331\times10^{-15}$\\
$A$ & $8.327\times10^{-17}$ & $6.939\times10^{-17}$ & $2.776\times10^{-17}$\\
$C$ & $1.110\times10^{-16}$ & $1.110\times10^{-16}$ & $4.163\times10^{-17}$\\
$D$ & $8.882\times10^{-16}$ & $5.274\times10^{-16}$ & $5.274\times10^{-16}$\\
$H$ & $2.665\times10^{-15}$ & $2.220\times10^{-15}$ & $2.331\times10^{-15}$\\
\bottomrule\end{tabular*}\end{table}

\section{Parameter comparisons and analytical checks}
Tables~\ref{tab:s10}--\ref{tab:s12} use the 2692-node mesh and the 16 current columns of main-text section 3.2. The eight conditions comprise five uniform impedance ratios $z/z_0\in\{0.01,0.1,1,10,100\}$, two inclusion conductivities $\sigma_{\rm inc}\in\{0.1,10\}$ and the non-uniform impedances $z_\ell=z_0[1+0.5\sin(2\pi(\ell-1)/16)]$, with $z_0=0.125$; each changes only the stated parameter from the baseline model. For every condition, $R$ is the EIDORS forward solution and $A$ the MATLAB Robin solution computed from the EIDORS matrices of that condition, and table~\ref{tab:s12} adds NGSolve's $B_s$ and $C_s$ at $z/z_0=0.01$ and $100$. All results are directly computed. Table~\ref{tab:s13} reports the two-electrode rectangular case.

\begin{table}[H]\centering\small
\caption{Parameter comparisons between the EIDORS augmented implementation $R$ and the MATLAB Robin implementation $A$ using the same assembled matrices. Relative differences are normalised by the norm of the EIDORS electrode potentials for the corresponding condition.}\label{tab:s10}
\begin{tabular*}{\textwidth}{@{\extracolsep{\fill}}lllll@{}}
\toprule
Condition & \makecell[l]{Maximum absolute\\electrode difference} & \makecell[l]{Relative electrode-\\potential difference} & $\eta_\infty(A)$ & \makecell[l]{Electrode-equation\\residual}\\\midrule
$\sigma_{\rm inc}=10$ & $5.329\times10^{-15}$ & $1.850\times10^{-15}$ & $1.709\times10^{-17}$ & $7.026\times10^{-16}$\\
$\sigma_{\rm inc}=0.1$ & $3.020\times10^{-14}$ & $2.079\times10^{-14}$ & $1.542\times10^{-16}$ & $5.951\times10^{-16}$\\
$z/z_0=0.01$ & $1.263\times10^{-13}$ & $1.002\times10^{-13}$ & $1.391\times10^{-16}$ & $3.659\times10^{-14}$\\
$z/z_0=0.1$ & $1.510\times10^{-14}$ & $1.042\times10^{-14}$ & $1.033\times10^{-16}$ & $4.237\times10^{-15}$\\
$z/z_0=1$ & $1.243\times10^{-14}$ & $5.246\times10^{-15}$ & $7.364\times10^{-17}$ & $5.699\times10^{-16}$\\
$z/z_0=10$ & $2.398\times10^{-14}$ & $2.285\times10^{-15}$ & $2.165\times10^{-17}$ & $4.431\times10^{-16}$\\
$z/z_0=100$ & $5.684\times10^{-14}$ & $4.595\times10^{-16}$ & $2.938\times10^{-18}$ & $5.035\times10^{-16}$\\
Non-uniform $z_\ell$ & $9.770\times10^{-15}$ & $4.376\times10^{-15}$ & $6.414\times10^{-17}$ & $6.468\times10^{-16}$\\
\bottomrule\end{tabular*}\end{table}

\begin{table}[H]\centering\small
\caption{Matrix condition numbers and maximum absolute ground-row reactions of the EIDORS benchmark solutions for each parameter condition. Condition numbers are computed as in table~\ref{tab:s2}. The reaction is the left-hand side of the deleted ground-node equation evaluated at the benchmark solution, that is, the discrete current required to maintain the grounding constraint (main-text section 2.4).}\label{tab:s11}
\begin{tabular*}{\textwidth}{@{\extracolsep{\fill}}llll@{}}
\toprule
Condition & $\widehat\kappa_\infty(H_g)$ & $\kappa_\infty(M_g)$ & Ground-row reaction\\\midrule
$\sigma_{\rm inc}=10$ & $1.493\times10^{4}$ & $9.008\times10^{0}$ & $1.454\times10^{-14}$\\
$\sigma_{\rm inc}=0.1$ & $3.526\times10^{3}$ & $9.012\times10^{1}$ & $1.063\times10^{-14}$\\
$z/z_0=0.01$ & $3.765\times10^{3}$ & $2.949\times10^{1}$ & $7.185\times10^{-14}$\\
$z/z_0=0.1$ & $2.702\times10^{3}$ & $2.608\times10^{1}$ & $1.225\times10^{-14}$\\
$z/z_0=1$ & $3.754\times10^{3}$ & $1.484\times10^{1}$ & $9.400\times10^{-15}$\\
$z/z_0=10$ & $1.092\times10^{4}$ & $3.769\times10^{0}$ & $1.791\times10^{-14}$\\
$z/z_0=100$ & $2.320\times10^{4}$ & $1.307\times10^{0}$ & $1.027\times10^{-14}$\\
Non-uniform $z_\ell$ & $3.765\times10^{3}$ & $1.924\times10^{1}$ & $1.357\times10^{-14}$\\
\bottomrule\end{tabular*}\end{table}

\begin{table}[H]\centering\small
\caption{NGSolve comparisons at the two endpoint impedance settings. Implementation $B_s$ reuses the EIDORS benchmark matrices, whereas $C_s$ uses its own assembled matrices with local integrals recorded. Both are compared with the directly computed EIDORS result for the corresponding condition.}\label{tab:s12}
\begin{tabular*}{\textwidth}{@{\extracolsep{\fill}}llll@{}}
\toprule
Condition/implementation & \makecell[l]{Maximum absolute\\electrode difference} & \makecell[l]{Relative electrode-\\potential difference} & $\eta_\infty$\\\midrule
$z/z_0=0.01$ / $B_s$ & $1.379\times10^{-13}$ & $9.783\times10^{-14}$ & $1.187\times10^{-16}$\\
$z/z_0=0.01$ / $C_s$ & $1.712\times10^{-13}$ & $1.003\times10^{-13}$ & $1.257\times10^{-16}$\\
$z/z_0=100$ / $B_s$ & $7.105\times10^{-14}$ & $9.059\times10^{-16}$ & $2.822\times10^{-18}$\\
$z/z_0=100$ / $C_s$ & $7.105\times10^{-14}$ & $7.580\times10^{-16}$ & $3.438\times10^{-18}$\\
\bottomrule\end{tabular*}\end{table}

For the rectangular case, the analytical expressions of the main text, with the same nodal grounding, are evaluated exactly at the stored binary64 coordinates and rounded to give $\bar{\bm{u}}_{\rm ana}$ and $\bar{\bm{U}}_{\rm ana}$. The relative differences in table~\ref{tab:s13} are normalised by this analytical reference, $e_u^{\rm ana}=\|\bm{u}_s-\bar{\bm{u}}_{\rm ana}\|_2/\|\bar{\bm{u}}_{\rm ana}\|_2$, and analogously for the electrode potentials.
\begin{table}[H]\centering\small
\caption{Comparison of the EIDORS augmented implementation $R$ and the MATLAB Robin implementation $A$ with the analytical reference for the rectangular domain.}\label{tab:s13}
\begin{tabular*}{\textwidth}{@{\extracolsep{\fill}}lllll@{}}
\toprule
Implementation & \makecell[l]{Maximum absolute\\nodal difference} & $e_u^{\rm ana}$ & \makecell[l]{Maximum absolute\\electrode difference} & $e_U^{\rm ana}$\\\midrule
$R$ & $6.661\times10^{-14}$ & $6.428\times10^{-14}$ & $6.550\times10^{-14}$ & $4.788\times10^{-14}$\\
$A$ & $7.994\times10^{-15}$ & $6.990\times10^{-15}$ & $7.550\times10^{-15}$ & $5.639\times10^{-15}$\\
\bottomrule\end{tabular*}\end{table}

\section{Software configurations and reproducibility}
\subsection{EIDORS benchmark calls and platform interfaces}
The EIDORS model uses the common nodal coordinates, element connectivity and electrode labels, with the potential at node $g=N$ fixed to zero and the identity matrix of order $L$ as the measurement pattern for each stimulation column. The entry point \texttt{fwd\_solve} is called with the first-order solver \path{fwd_solve_1st_order} and the assembly function \path{system_mat_1st_order}, which solve the conventional augmented CEM equations. The \texttt{get\_all\_nodes} option returns all nodal and CEM electrode potentials, the first $N$ rows containing $\bm{u}$ and the following $L$ rows $\bm{U}$.

The full output has dimensions $(N+16)\times16$ in the 16-electrode cases, with row $g$ identically zero, and $317\times1$ in the two-electrode rectangular case; after reordering, the identity-measurement output equals $\bm{U}$. In each 16-electrode case, the two EIDORS runs were made in separate MATLAB processes with cleared caches, and the calls to the assembly and linear-solver functions were recorded. The MATLAB Robin solution $A$ was computed in the same runs from the extracted matrices. Each parameter condition and the rectangular case used a single EIDORS calculation and a single MATLAB Robin calculation.

\noindent\textbf{Software versions and solver configurations.} The EIDORS benchmark uses MATLAB R2023b Update 4 and the first-order forward solver of EIDORS 3.12, and implementation $A$ uses MATLAB LU factorisation.

DOLFINx 0.10.0.post1 uses dense LAPACK LU through PETSc; FreeFEM 4.9 uses the UMFPACK interface of \texttt{sparsesolver}; GetFEM 5.3 uses MUMPS; and MFEM 4.9 uses \texttt{DenseMatrixInverse} with LAPACK enabled. On these four platforms, NumPy forms the electrode transadmittance matrix and recovers the nodal potentials. NGSolve 6.2.2606 uses \texttt{Matrix.I} and the platform's matrix-product operations. Bitwise repeatability refers to these versions and configurations.

\subsection{Platform-native Robin finite element assembly and comparison settings}
In the platform-native implementation $P_s$, each platform receives the domain diffusion weak form, the Robin boundary terms and the loads over entire electrodes, and its native finite element integration and assembly interfaces evaluate the integrals and accumulate the contributions according to mesh connectivity. The imported data are checked against the element and boundary indices, and material and electrode regions follow the physical model. NGSolve, DOLFINx and FreeFEM assemble global boundary weak forms, MFEM calls boundary integrators, and GetFEM calls \path{add_Fourier_Robin_brick}. Each unit-potential load column is assembled once over the entire corresponding electrode. All platforms were used without modification of their library source code.

Let $\beta_\ell=1/z_\ell>0$. All five platforms add $+\int_{E_\ell}\beta_\ell u\psi\,\mathrm{d}s$ to the left-hand side and use $+\int_{E_\ell}\beta_\ell \psi\,\mathrm{d}s$ on the right-hand side for a unit electrode-potential response, where $\psi$ is the test function. These terms are absent on the insulating gaps, and all boundary integrals are evaluated over the entire electrode $E_\ell$.

\noindent\textbf{Robin boundary terms and public interfaces.} NGSolve assembles the boundary terms and loads using \texttt{BilinearForm} and \texttt{LinearForm}, and computes $D$ using \texttt{Integrate}. DOLFINx passes UFL forms to \path{assemble_matrix} and \path{assemble_vector}, and computes $D$ using \path{assemble_scalar}.

MFEM's boundary \texttt{MassIntegrator} and \texttt{BoundaryLFIntegrator} form the Robin bilinear term and boundary load, respectively. FreeFEM evaluates these boundary integrals and $D$ using \texttt{int1d} within a global \texttt{varf}. GetFEM uses \path{add_Fourier_Robin_brick} and \path{add_source_term_brick} for these two terms, respectively, and computes $D$ using \path{asm_generic}.
Both the Robin and source-term bricks in GetFEM use positive $\beta_\ell$ for the left-hand-side term and right-hand-side load specified above. For each column of $C$, only the source term on the corresponding electrode is enabled. Each platform first reorders the assembled matrices to the common degree-of-freedom ordering. The ground-node row and column are then deleted from $H$, and the corresponding row is deleted from $C$, yielding $H_g$ and $C_g$, respectively. For the 2692-node case, $H_g\in\mathbb{R}^{2691\times2691}$ and $C_g\in\mathbb{R}^{2691\times16}$. For the 2535-node case, their dimensions are $2534\times2534$ and $2534\times16$, respectively. In both cases, $D\in\mathbb{R}^{16\times16}$ retains all electrodes.

Each platform's integration interface evaluates $D_{\ell\ell}=\int_{E_\ell}z_\ell^{-1}\,\mathrm{d}s$. MFEM evaluates this integral by applying its boundary linear form to the all-ones vector, using the partition-of-unity property of the P1 basis functions. Its rounding order can therefore differ from that of edgewise scalar integration. NGSolve, MFEM and FreeFEM use the default quadrature rules of the listed interfaces; DOLFINx lets UFL/FFCx determine the polynomial quadrature degree automatically; and GetFEM uses the order-4 rule in \texttt{MeshIm}. These differences are retained in the platform-native results for each mesh.

The local-integral-recording implementation $C_s$ extracts the local matrices of every element and electrode edge and accumulates part of them in a common order, so that intermediate values are available for the analyses of sections 3, 4 and 7, whereas $B_s$ reuses the EIDORS blocks and $P_s$ uses platform-native assembly. Within each platform, $P_s$ and $C_s$ use the same solver interfaces. With these public interfaces, the electrode integrals and the electrode system must be specified explicitly, but no dedicated CEM forward-solver entry point is required.

\subsection{Refinement and repeatability}
In the 2692-node case, refinement is applied to the grounded augmented system for $R$ and to the response equations $H_gZ=C_g$ for $A$, $B_s$ and $C_s$; the platform-native results $P_s$ are not refined. The procedure is that of main-text section 3.5: the stored matrix entries are fixed, residuals and updates are evaluated in rational arithmetic, and each correction is computed in double precision with a SciPy SuperLU factorisation of the implementation's own matrix. After steps 3 and 4, the responses are used to form $M_g$, solve for $\bm{U}$ and recover the nodal potentials in rational arithmetic, and the double-precision roundings of the two results are compared. The 2535-node refinement results are given in section 9.

Each implementation in the two 16-electrode cases was run twice from the same fixed inputs. In the 2692-node case, the comparisons covered the assembled matrices and directly computed potentials of $R$, $A$ and every $P_s$, $B_s$ and $C_s$, together with the response and transadmittance matrices of the Robin implementations, the recorded local matrices of $C_s$ and the double-precision roundings of the refined potentials of $R$, $A$, $B_s$ and $C_s$. In the 2535-node case, the comparisons covered the assembled matrices and directly computed potentials of $R$, $A$ and the five $P_s$ implementations, together with the response and transadmittance matrices of the Robin implementations. All compared quantities were bitwise identical between the two runs.

The input mappings of the 2692-node case were also checked with a single C program that performs P1 assembly and the Robin--transadmittance solution with a fixed sequence of integration, accumulation, LU factorisation and potential recovery. The coordinates, connectivity, material values, electrode labels and currents imported by EIDORS and by each platform were restored to the common ordering and passed to this program. At every stage, its intermediate matrices and final potentials were bitwise identical across all sets of imported inputs and both runs. The reproduction package lists the programs and output files that produce each table and figure.

\section{Platform-native and local-integral-recording assembly}
Table~\ref{tab:s14} compares the shared-matrix implementations $B_s$, which reuse the EIDORS matrices, with the local-integral-recording implementations $C_s$. Table~\ref{tab:s15} extends main-text table 3 to all nodal potentials and adds the electrode-equation residuals of the platform-native implementations $P_s$, evaluated with their own matrices. All results in this section refer to the 2692-node case.
\begin{table}[H]\centering\small
\caption{Comparison of the Robin implementations used for error analysis. The first two differences use directly computed results; the last compares the refined approximations of $C_s$ and $B_s$, whose assembled matrices differ. Maxima are taken over all 16 stimulation columns and all 16 electrodes.}\label{tab:s14}
\begin{tabular*}{\textwidth}{@{\extracolsep{\fill}}llll@{}}
\toprule
Platform & $\|\bm{U}_{B_s}-\bm{U}_A\|_{\max}$ & $\|\bm{U}_{C_s}-\bm{U}_R\|_{\max}$ & $\|\bm{U}_{C_s}^\star-\bm{U}_{B_s}^\star\|_{\max}$\\\midrule
DOLFINx & $3.642\times10^{-14}$ & $4.841\times10^{-14}$ & $7.011\times10^{-15}$\\
FreeFEM & $9.104\times10^{-15}$ & $1.421\times10^{-14}$ & $1.262\times10^{-14}$\\
GetFEM & $1.110\times10^{-14}$ & $2.931\times10^{-14}$ & $1.508\times10^{-14}$\\
MFEM & $4.632\times10^{-13}$ & $4.641\times10^{-13}$ & $8.867\times10^{-15}$\\
NGSolve & $3.464\times10^{-14}$ & $2.620\times10^{-14}$ & $1.028\times10^{-14}$\\
\bottomrule\end{tabular*}\end{table}

\begin{table}[H]\centering\small
\caption{Comparison of each platform-native implementation $P_s$ with the EIDORS benchmark $R$, using directly computed results. Electrode-equation residuals are evaluated using the assembled matrices of $P_s$.}\label{tab:s15}
\footnotesize\setlength{\tabcolsep}{3pt}
\begin{tabular*}{\textwidth}{@{\extracolsep{\fill}}llllll@{}}
\toprule
Platform & \makecell[l]{Maximum absolute\\nodal difference} & $e_u$ & \makecell[l]{Maximum absolute\\electrode difference} & $e_U$ & \makecell[l]{Electrode-equation\\residual}\\\midrule
DOLFINx & $3.508\times10^{-14}$ & $5.125\times10^{-14}$ & $3.531\times10^{-14}$ & $2.539\times10^{-14}$ & $3.138\times10^{-16}$\\
FreeFEM & $2.176\times10^{-14}$ & $1.943\times10^{-14}$ & $1.910\times10^{-14}$ & $9.165\times10^{-15}$ & $7.257\times10^{-16}$\\
GetFEM & $3.308\times10^{-14}$ & $1.885\times10^{-14}$ & $2.842\times10^{-14}$ & $9.164\times10^{-15}$ & $6.469\times10^{-16}$\\
MFEM & $4.676\times10^{-13}$ & $6.289\times10^{-13}$ & $4.641\times10^{-13}$ & $3.069\times10^{-13}$ & $4.049\times10^{-16}$\\
NGSolve & $1.266\times10^{-14}$ & $8.653\times10^{-15}$ & $1.732\times10^{-14}$ & $7.539\times10^{-15}$ & $5.258\times10^{-16}$\\
\bottomrule\end{tabular*}\end{table}

Within each platform, $P_s$ and $C_s$ use the same solver interfaces but differ in how local contributions are accumulated and in some quadrature rules. Table~\ref{tab:s16} compares the matrices $H_g$, $C_g$ and $D$ passed to their solvers; the resulting potential differences include the propagation of these assembly changes through the floating-point solution.
\begin{table}[H]\centering\small
\caption{Platform-native assembly $P_s$ versus local-integral-recording assembly $C_s$: maximum absolute differences in solver-input matrices and directly computed electrode potentials. The first three columns report maximum absolute entrywise differences between the respective assembled matrices $H_g,C_g,D$.}\label{tab:s16}
\begin{tabular*}{\textwidth}{@{\extracolsep{\fill}}lllll@{}}
\toprule
Platform & $\Delta H_g$ & $\Delta C_g$ & $\Delta D$ & $\|\bm{U}_{P_s}-\bm{U}_{C_s}\|_{\max}$\\\midrule
DOLFINx & $4.441\times10^{-15}$ & $5.551\times10^{-17}$ & $2.220\times10^{-16}$ & $1.910\times10^{-14}$\\
FreeFEM & $4.441\times10^{-16}$ & $2.776\times10^{-17}$ & $6.661\times10^{-16}$ & $6.661\times10^{-15}$\\
GetFEM & $1.776\times10^{-15}$ & $0$ & $6.661\times10^{-16}$ & $3.109\times10^{-15}$\\
MFEM & $4.441\times10^{-16}$ & $0$ & $2.220\times10^{-16}$ & $3.997\times10^{-15}$\\
NGSolve & $4.441\times10^{-16}$ & $0$ & $4.885\times10^{-15}$ & $3.064\times10^{-14}$\\
\bottomrule\end{tabular*}\end{table}

\clearpage
\section{Response-equation residuals with shared assembled matrices}
Figure~\ref{fig:shared_response} and table~\ref{tab:response_history} show the response-equation residual $\|C_g-H_gZ^{(k)}\|_{\max}$ of the implementations that share the EIDORS blocks $H_g$, $C_g$ and $D$ in the 2692-node case. The MATLAB implementation $A$ and each $B_s$ start from their own double-precision response ($k=0$) and are refined as in main-text section 3.5. For every implementation, the residual decreases by roughly 14--15 orders of magnitude per step.

\begin{figure}[H]\centering
\includegraphics[width=0.95\linewidth]{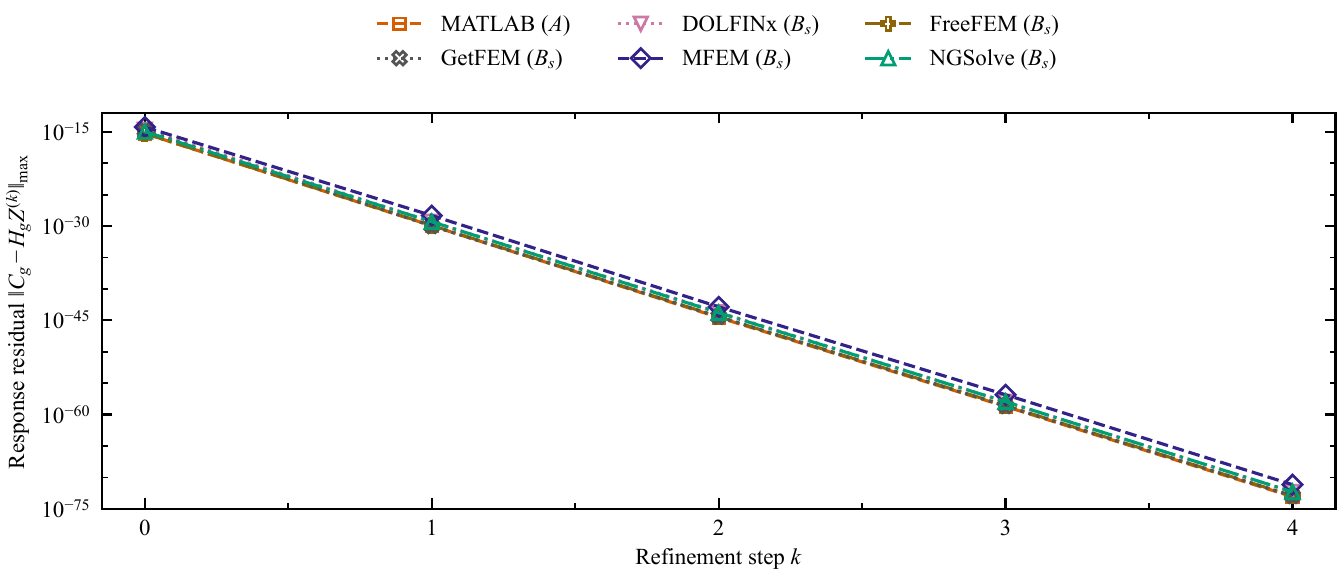}
\caption{Response-equation residuals of the Robin implementations that share the EIDORS matrices in the 2692-node case, during iterative refinement with a separate SuperLU factorisation for each implementation.}\label{fig:shared_response}
\end{figure}

\begin{table}[H]\centering\small
\caption{Data for figure~\ref{fig:shared_response}: Robin response-equation residuals at successive iterative refinement steps when reusing the EIDORS benchmark matrices. Values are shown to four significant figures.}\label{tab:s17}
\footnotesize\setlength{\tabcolsep}{3pt}
\label{tab:response_history}\begin{tabular*}{\textwidth}{@{\extracolsep{\fill}}llllll@{}}
\toprule
Initial response source & $k=0$ & $k=1$ & $k=2$ & $k=3$ & $k=4$\\\midrule
MATLAB Robin ($A$) & $5.186\times10^{-16}$ & $1.097\times10^{-30}$ & $2.847\times10^{-45}$ & $1.958\times10^{-59}$ & $8.585\times10^{-74}$\\
DOLFINx ($B_s$) & $2.136\times10^{-15}$ & $4.881\times10^{-30}$ & $1.900\times10^{-44}$ & $1.078\times10^{-58}$ & $4.849\times10^{-73}$\\
FreeFEM ($B_s$) & $5.407\times10^{-16}$ & $1.245\times10^{-30}$ & $3.723\times10^{-45}$ & $2.696\times10^{-59}$ & $1.171\times10^{-73}$\\
GetFEM ($B_s$) & $1.552\times10^{-15}$ & $1.019\times10^{-30}$ & $3.785\times10^{-45}$ & $2.342\times10^{-59}$ & $1.155\times10^{-73}$\\
MFEM ($B_s$) & $5.620\times10^{-15}$ & $4.661\times10^{-29}$ & $1.390\times10^{-43}$ & $1.432\times10^{-57}$ & $7.144\times10^{-72}$\\
NGSolve ($B_s$) & $1.058\times10^{-15}$ & $4.776\times10^{-30}$ & $1.565\times10^{-44}$ & $1.103\times10^{-58}$ & $5.043\times10^{-73}$\\
\bottomrule\end{tabular*}\end{table}

\clearpage\section{The 2535-node exact rational-arithmetic test case}

This section gives the details and full results of the 2535-node exact discrete benchmark, and table~\ref{tab:s27} adds a 1055-node case. The directly computed results of $R$, $A$ and the five $P_s$ implementations are the starting points ($k=0$) for four refinement steps.

\begin{table}[H]\centering\small
\caption{Specification of the 2535-node rational-arithmetic benchmark.}\label{tab:s18}
\begin{tabularx}{\textwidth}{@{}>{\raggedright\arraybackslash}p{0.29\textwidth}Y@{}}
\toprule
Item & Specification or result\\\midrule
Mesh and electrodes & 2535 nodes, 4940 elements, 128 boundary edges, 16 electrodes\\
Coordinates & $p/65536$; imported values identical to the exact inputs\\
Electrode-edge lengths & All 64 lengths rational (integer Pythagorean triples)\\
Minimum angle & $28.121637^\circ$\\
Conductivity and contact impedance & $\sigma\in\{1,2\}$, $z_\ell=1/8$\\
Currents and potential reference & 16 current-input columns, exact rank 15; node 2535 at the origin is grounded\\
Grounded transadmittance matrix & $16\times16$, exactly symmetric, rank 16\\
Independent assembly and equation checks & The two exact assemblies give identical matrices; every entry of the full augmented-system residual is exactly zero\\
Reference output & 40816 exact fractions and their correctly rounded values\\
\bottomrule\end{tabularx}\end{table}

The P1 stiffness matrices are computed exactly from the rational coordinates, area determinants and piecewise-constant conductivities, and the electrode terms are integrated exactly over the straight electrode edges, whose lengths are rational. The boundary-node radii range from approximately 0.999707 to 1.001216. Unlike the fitted interface of the 2692-node mesh, the inclusion is defined by element labels assigned with an exact centroid test. The two independent exact assemblies compute the basis-function gradients from cofactors and from the inverse of the matrix defining the affine basis functions, respectively, and the boundary mass terms are additionally evaluated by polynomial integration.

Exact arithmetic uses python-FLINT 0.9.0 and the Dixon solver of FLINT 3.6.0. The grounded nodal system of order 2534 is solved for the 16 response columns simultaneously, after which the grounded transadmittance matrix is formed, the electrode system solved and the nodal potentials recovered. Numerators and denominators of the reference contain at most 32761 bits. The reference satisfies the grounded system with zero residual; the exactly assembled matrices are symmetric and satisfy the constant-nullspace identities; the current inputs are conservative; and, with the ground value set to zero, the solution also satisfies the deleted ground-row equation. Positive conductivity, positive contact impedance and mesh connectivity make the grounded system positive definite. The correctly rounded double-precision values of the reference are obtained by comparison with the exact midpoints between adjacent binary64 numbers.

\begin{table}[H]\centering\small
\caption{Differences of the directly computed results from the EIDORS output in the 2535-node case.}\label{tab:s19}
\begin{tabular*}{\textwidth}{@{\extracolsep{\fill}}lll@{}}
\toprule
Implementation & \makecell[l]{Maximum absolute\\nodal difference} & \makecell[l]{Maximum absolute\\electrode difference}\\\midrule
EIDORS augmented ($R$) & $0$ & $0$\\
MATLAB Robin ($A$) & $3.642\times10^{-14}$ & $3.531\times10^{-14}$\\
DOLFINx ($P_s$) & $1.003\times10^{-13}$ & $9.925\times10^{-14}$\\
FreeFEM ($P_s$) & $3.730\times10^{-14}$ & $3.708\times10^{-14}$\\
GetFEM ($P_s$) & $3.819\times10^{-14}$ & $3.642\times10^{-14}$\\
MFEM ($P_s$) & $8.282\times10^{-14}$ & $8.049\times10^{-14}$\\
NGSolve ($P_s$) & $1.237\times10^{-13}$ & $1.275\times10^{-13}$\\
\bottomrule\end{tabular*}\end{table}

\begingroup
\setlength{\intextsep}{4pt plus 1pt minus 1pt}
\renewcommand{\arraystretch}{1.04}
\begin{table}[H]\centering\small
\caption{Maximum absolute entrywise errors of the finite element blocks used by each implementation, relative to exact assembly. Only $H_g$, $C_g$ and $D$, which are available from all interfaces, are compared.}\label{tab:s20}
\begin{tabular*}{\textwidth}{@{\extracolsep{\fill}}llll@{}}
\toprule
Implementation & $H_g$ & $C_g$ & $D$\\\midrule
EIDORS augmented ($R$) & $1.513\times10^{-15}$ & $2.220\times10^{-16}$ & $4.441\times10^{-16}$\\
MATLAB Robin ($A$) & $1.513\times10^{-15}$ & $2.220\times10^{-16}$ & $4.441\times10^{-16}$\\
DOLFINx ($P_s$) & $1.153\times10^{-15}$ & $5.551\times10^{-17}$ & $0$\\
FreeFEM ($P_s$) & $1.907\times10^{-15}$ & $0$ & $2.220\times10^{-16}$\\
GetFEM ($P_s$) & $2.691\times10^{-14}$ & $7.772\times10^{-16}$ & $3.109\times10^{-15}$\\
MFEM ($P_s$) & $8.385\times10^{-16}$ & $0$ & $0$\\
NGSolve ($P_s$) & $8.385\times10^{-16}$ & $0$ & $4.441\times10^{-15}$\\
\bottomrule\end{tabular*}\end{table}

\begin{table}[H]\centering\small
\caption{Maximum absolute entrywise errors in the electrode transadmittance matrices computed by the Robin implementations, relative to $M_{g,\mathbb Q}$. These errors include contributions from finite element assembly, the response solution and, for $k=0$, the floating-point formation of $M_g$. $R$ is not listed because the EIDORS solver does not form this matrix.}\label{tab:s21}
\begin{tabular*}{\textwidth}{@{\extracolsep{\fill}}lll@{}}
\toprule
Implementation & $k=0$ & $k=4$\\\midrule
MATLAB Robin ($A$) & $1.339\times10^{-14}$ & $1.321\times10^{-14}$\\
DOLFINx ($P_s$) & $4.694\times10^{-15}$ & $4.046\times10^{-15}$\\
FreeFEM ($P_s$) & $6.976\times10^{-15}$ & $6.916\times10^{-15}$\\
GetFEM ($P_s$) & $6.889\times10^{-15}$ & $9.481\times10^{-15}$\\
MFEM ($P_s$) & $1.411\times10^{-14}$ & $3.663\times10^{-15}$\\
NGSolve ($P_s$) & $8.462\times10^{-15}$ & $7.962\times10^{-15}$\\
\bottomrule\end{tabular*}\end{table}

\begin{table}[H]\centering\small
\caption{Maximum absolute electrode-potential errors relative to the exact discrete reference at each refinement step, over all 16 stimulation columns. These data correspond to figure 4(b) of the main text and figure~\ref{fig:errors_individual}.}\label{tab:s22}
\footnotesize\setlength{\tabcolsep}{3pt}
\begin{tabular*}{\textwidth}{@{\extracolsep{\fill}}llllll@{}}
\toprule
Implementation & $k=0$ & $k=1$ & $k=2$ & $k=3$ & $k=4$\\\midrule
EIDORS augmented ($R$) & $7.809\times10^{-14}$ & $1.086\times10^{-13}$ & $1.086\times10^{-13}$ & $1.086\times10^{-13}$ & $1.086\times10^{-13}$\\
MATLAB Robin ($A$) & $1.087\times10^{-13}$ & $1.086\times10^{-13}$ & $1.086\times10^{-13}$ & $1.086\times10^{-13}$ & $1.086\times10^{-13}$\\
DOLFINx ($P_s$) & $3.371\times10^{-14}$ & $3.092\times10^{-14}$ & $3.092\times10^{-14}$ & $3.092\times10^{-14}$ & $3.092\times10^{-14}$\\
FreeFEM ($P_s$) & $5.402\times10^{-14}$ & $5.585\times10^{-14}$ & $5.585\times10^{-14}$ & $5.585\times10^{-14}$ & $5.585\times10^{-14}$\\
GetFEM ($P_s$) & $4.558\times10^{-14}$ & $8.072\times10^{-14}$ & $8.072\times10^{-14}$ & $8.072\times10^{-14}$ & $8.072\times10^{-14}$\\
MFEM ($P_s$) & $1.014\times10^{-13}$ & $2.994\times10^{-14}$ & $2.994\times10^{-14}$ & $2.994\times10^{-14}$ & $2.994\times10^{-14}$\\
NGSolve ($P_s$) & $5.323\times10^{-14}$ & $4.427\times10^{-14}$ & $4.427\times10^{-14}$ & $4.427\times10^{-14}$ & $4.427\times10^{-14}$\\
\bottomrule\end{tabular*}\end{table}

\begin{table}[H]\centering\small
\caption{Maximum absolute nodal-potential errors relative to the exact discrete reference at each refinement step, over all 16 stimulation columns.}\label{tab:s23}
\footnotesize\setlength{\tabcolsep}{3pt}
\begin{tabular*}{\textwidth}{@{\extracolsep{\fill}}llllll@{}}
\toprule
Implementation & $k=0$ & $k=1$ & $k=2$ & $k=3$ & $k=4$\\\midrule
EIDORS augmented ($R$) & $7.875\times10^{-14}$ & $1.093\times10^{-13}$ & $1.093\times10^{-13}$ & $1.093\times10^{-13}$ & $1.093\times10^{-13}$\\
MATLAB Robin ($A$) & $1.094\times10^{-13}$ & $1.093\times10^{-13}$ & $1.093\times10^{-13}$ & $1.093\times10^{-13}$ & $1.093\times10^{-13}$\\
DOLFINx ($P_s$) & $3.380\times10^{-14}$ & $3.096\times10^{-14}$ & $3.096\times10^{-14}$ & $3.096\times10^{-14}$ & $3.096\times10^{-14}$\\
FreeFEM ($P_s$) & $5.456\times10^{-14}$ & $5.628\times10^{-14}$ & $5.628\times10^{-14}$ & $5.628\times10^{-14}$ & $5.628\times10^{-14}$\\
GetFEM ($P_s$) & $4.523\times10^{-14}$ & $8.118\times10^{-14}$ & $8.118\times10^{-14}$ & $8.118\times10^{-14}$ & $8.118\times10^{-14}$\\
MFEM ($P_s$) & $1.045\times10^{-13}$ & $3.031\times10^{-14}$ & $3.031\times10^{-14}$ & $3.031\times10^{-14}$ & $3.031\times10^{-14}$\\
NGSolve ($P_s$) & $4.893\times10^{-14}$ & $3.983\times10^{-14}$ & $3.983\times10^{-14}$ & $3.983\times10^{-14}$ & $3.983\times10^{-14}$\\
\bottomrule\end{tabular*}\end{table}

\begin{table}[H]\centering\small
\caption{Grounded augmented-system residuals $\|\bm{b}-\widehat F_{g,s}\bm{x}_{g,s}^{(k)}\|_{\max}$ at each refinement step, excluding the deleted ground row. Each implementation uses its own assembled matrix. These data correspond to figure 4(a) of the main text and figure~\ref{fig:rational_residual}.}\label{tab:s24}
\footnotesize\setlength{\tabcolsep}{3pt}
\begin{tabular*}{\textwidth}{@{\extracolsep{\fill}}llllll@{}}
\toprule
Implementation & $k=0$ & $k=1$ & $k=2$ & $k=3$ & $k=4$\\\midrule
EIDORS augmented ($R$) & $2.753\times10^{-15}$ & $9.773\times10^{-29}$ & $1.236\times10^{-42}$ & $2.732\times10^{-56}$ & $3.224\times10^{-70}$\\
MATLAB Robin ($A$) & $1.450\times10^{-15}$ & $2.643\times10^{-30}$ & $9.765\times10^{-45}$ & $8.920\times10^{-59}$ & $6.461\times10^{-73}$\\
DOLFINx ($P_s$) & $2.237\times10^{-15}$ & $2.647\times10^{-30}$ & $1.637\times10^{-44}$ & $1.461\times10^{-58}$ & $1.540\times10^{-72}$\\
FreeFEM ($P_s$) & $1.688\times10^{-15}$ & $2.090\times10^{-30}$ & $5.726\times10^{-45}$ & $5.967\times10^{-59}$ & $3.805\times10^{-73}$\\
GetFEM ($P_s$) & $1.783\times10^{-15}$ & $3.844\times10^{-29}$ & $2.977\times10^{-43}$ & $2.185\times10^{-57}$ & $1.963\times10^{-71}$\\
MFEM ($P_s$) & $1.163\times10^{-14}$ & $5.741\times10^{-29}$ & $1.992\times10^{-43}$ & $1.850\times10^{-57}$ & $1.250\times10^{-71}$\\
NGSolve ($P_s$) & $2.268\times10^{-15}$ & $4.983\times10^{-30}$ & $1.548\times10^{-44}$ & $1.328\times10^{-58}$ & $9.805\times10^{-73}$\\
\bottomrule\end{tabular*}\end{table}
\endgroup

\clearpage
\begin{figure}[H]\centering
\includegraphics[width=\linewidth]{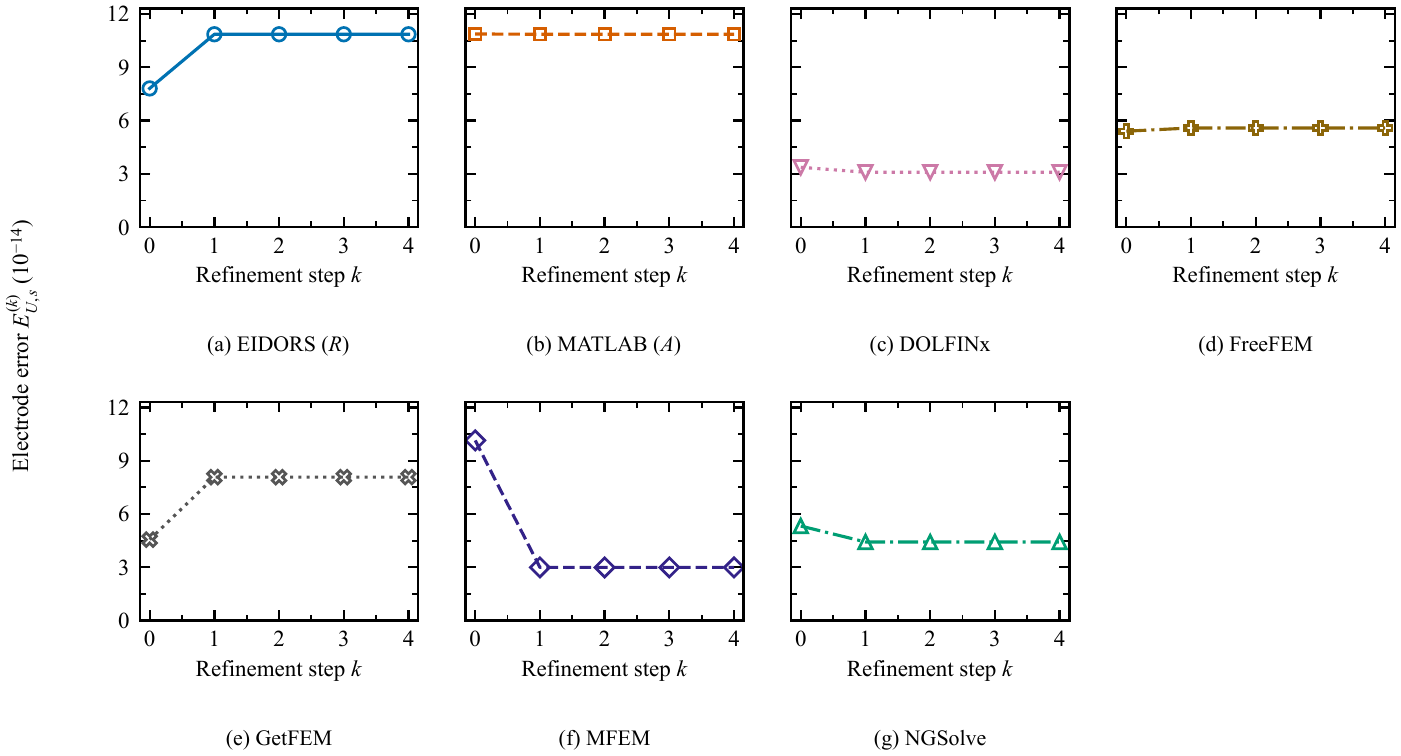}
\caption{Maximum absolute electrode-potential errors $E_{U,s}^{(k)}$ relative to the exact discrete reference for each implementation in the 2535-node case, plotted separately with common axis limits (in units of $10^{-14}$) to resolve curves that overlap in figure 4(b) of the main text. $k=0$ denotes the directly computed result; the data are listed in table~\ref{tab:s22}.}\label{fig:errors_individual}
\end{figure}

\begin{figure}[H]\centering
\includegraphics[width=0.93\linewidth]{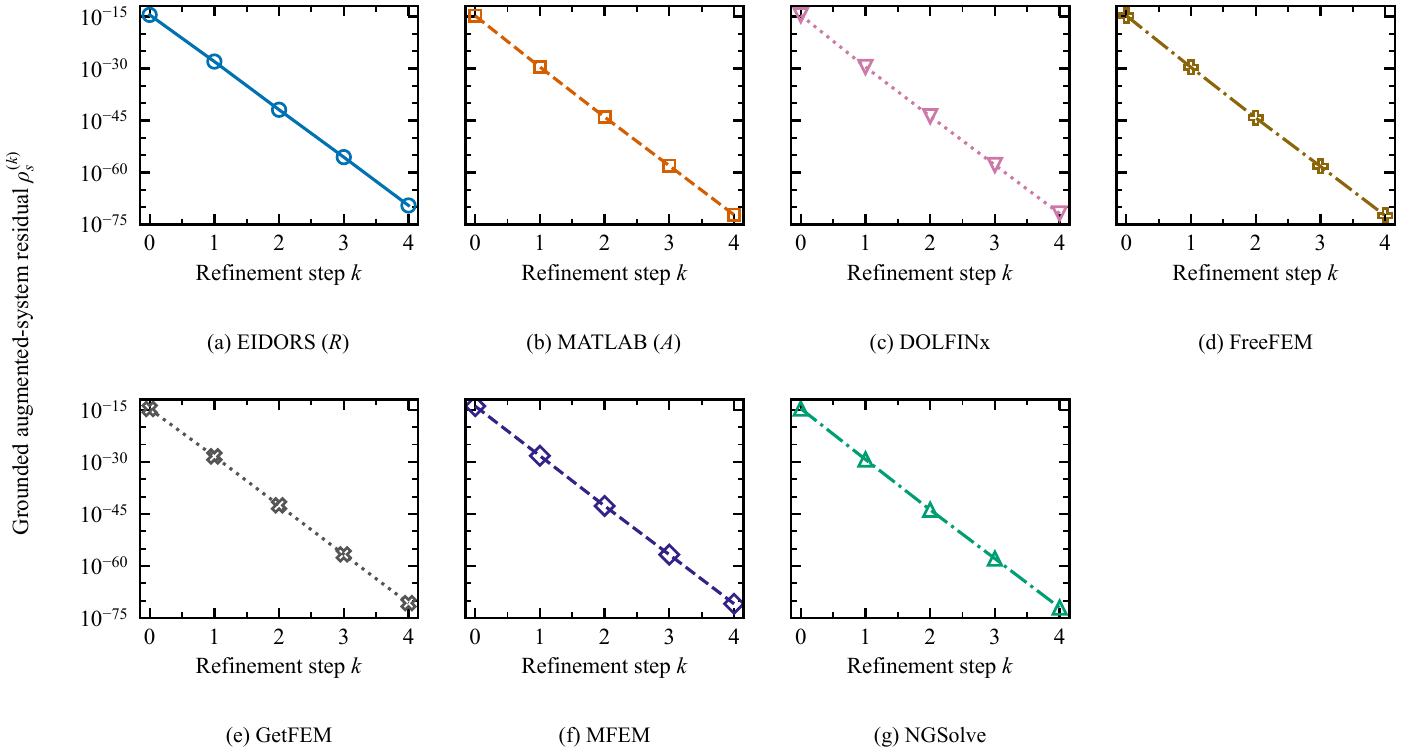}
\caption{Grounded augmented-system residuals of each implementation in the 2535-node case, evaluated with its own assembled matrix and plotted separately with common axis limits; the data are those of figure 4(a) of the main text and table~\ref{tab:s24}. The residuals decrease at every step, whereas the errors in figure~\ref{fig:errors_individual} remain essentially constant from $k=1$.}\label{fig:rational_residual}
\end{figure}

\clearpage\noindent Table~\ref{tab:s25} decomposes, for each implementation, the signed error at the electrode-potential component with the largest directly computed error. With $\widetilde{\bm{x}}_s$ the exact solution of the stored assembled system, $\bm{x}_s-\bm{x}_s^{(4)}$ estimates the solution error $\bm{x}_s-\widetilde{\bm{x}}_s$ and $\bm{x}_s^{(4)}-\bm{x}_{\mathbb Q}$ estimates the assembly-induced solution shift $\widetilde{\bm{x}}_s-\bm{x}_{\mathbb Q}$. The two estimates add up to the total error, and the refinement error remaining after four steps enters them with equal magnitude and opposite signs.

\begin{table}[H]\centering\small
\caption{Signed estimated error contributions at the electrode-potential component with the largest directly computed error of each implementation.}\label{tab:s25}
\footnotesize\setlength{\tabcolsep}{3pt}
\begin{tabular*}{\textwidth}{@{\extracolsep{\fill}}lllll@{}}
\toprule
Implementation & \makecell[l]{Electrode/\\stimulation} & Total error & \makecell[l]{Estimated solution-error\\contribution} & \makecell[l]{Estimated assembly\\contribution}\\\midrule
EIDORS augmented ($R$) & 8/9 & $-7.809\times10^{-14}$ & $3.048\times10^{-14}$ & $-1.086\times10^{-13}$\\
MATLAB Robin ($A$) & 8/9 & $-1.087\times10^{-13}$ & $-1.671\times10^{-16}$ & $-1.086\times10^{-13}$\\
DOLFINx ($P_s$) & 4/5 & $3.371\times10^{-14}$ & $5.118\times10^{-15}$ & $2.859\times10^{-14}$\\
FreeFEM ($P_s$) & 9/10 & $-5.402\times10^{-14}$ & $1.827\times10^{-15}$ & $-5.585\times10^{-14}$\\
GetFEM ($P_s$) & 9/10 & $-4.558\times10^{-14}$ & $3.514\times10^{-14}$ & $-8.072\times10^{-14}$\\
MFEM ($P_s$) & 8/9 & $-1.014\times10^{-13}$ & $-1.304\times10^{-13}$ & $2.899\times10^{-14}$\\
NGSolve ($P_s$) & 9/10 & $5.323\times10^{-14}$ & $8.962\times10^{-15}$ & $4.427\times10^{-14}$\\
\bottomrule\end{tabular*}\end{table}

\begin{table}[H]\centering\small
\caption{Numbers of values bitwise identical to the correctly rounded exact discrete reference, out of 40816 in every case. Counts include the 16 identically zero ground values.}\label{tab:s26}
\begin{tabular*}{\textwidth}{@{\extracolsep{\fill}}llll@{}}
\toprule
Implementation & Before refinement & After four steps & \makecell[l]{Stable rounding over\\the final two steps}\\\midrule
EIDORS augmented ($R$) & 23 & 20 & Yes\\
MATLAB Robin ($A$) & 20 & 20 & Yes\\
DOLFINx ($P_s$) & 63 & 57 & Yes\\
FreeFEM ($P_s$) & 28 & 27 & Yes\\
GetFEM ($P_s$) & 22 & 22 & Yes\\
MFEM ($P_s$) & 26 & 39 & Yes\\
NGSolve ($P_s$) & 62 & 24 & Yes\\
\bottomrule\end{tabular*}\end{table}

Because a single-bit difference counts as a disagreement, table~\ref{tab:s26} is a strict measure of binary agreement rather than of accuracy, which is quantified by the errors in tables~\ref{tab:s22} and \ref{tab:s23}.

\begin{table}[H]\centering\small
\caption{Rational-arithmetic benchmarks at two mesh sizes, each with its own exact discrete reference (zero residual, correctly rounded representation). Error ranges cover the seven implementations after four refinement steps.}\label{tab:s27}
\begin{tabular*}{\textwidth}{@{\extracolsep{\fill}}lllll@{}}
\toprule
Nodes & Elements & \makecell[l]{Number of\\exact outputs} & \makecell[l]{Minimum\\electrode error} & \makecell[l]{Maximum\\electrode error}\\\midrule
1055 & 1980 & 17136 & $9.600\times10^{-15}$ & $4.537\times10^{-14}$\\
2535 & 4940 & 40816 & $2.994\times10^{-14}$ & $1.086\times10^{-13}$\\
\bottomrule\end{tabular*}\end{table}

The two meshes share the boundary with its 128 nodes, the grounded origin, the electrodes, contact impedances, currents and material-assignment rule, but their interior triangulations are not nested. Exact assembly and solution and all seven implementations were run separately on each mesh, and each row is evaluated against the exact solution of its own discrete problem, so the table compares numerical errors rather than discretisation errors.

\noindent\textbf{Voltage-difference errors.} To separate voltage-difference errors from a common potential shift, table~\ref{tab:s28} reports $E_{\Delta U,s}^{(k)}=\|T_{16}(\bm{U}_s^{(k)}-\bm{U}_{\mathbb Q})\|_{\max}$, where $(T_{16}\bm{U})_i=U_i-U_{16}$ for $i=1,\ldots,15$, together with the analogous error for the adjacent-electrode voltages $(B\bm{U})_i=U_i-U_{i+1}$, with $U_{17}=U_1$. Since $T_{16}\mathbf1=B\mathbf1=0$, both measures are invariant under a common potential shift. They are evaluated exactly from the same electrode-error matrices $\bm{U}_s^{(k)}-\bm{U}_{\mathbb Q}$ as table~\ref{tab:s22}, with maxima taken over all 16 stimulation columns. The common-component column gives the largest magnitude, over stimulation columns, of the mean error of the 16 electrode potentials.

\begin{table}[H]\centering\small
\caption{Maximum absolute voltage-difference errors relative to the exact discrete reference for each implementation in the 2535-node case. Here, $k=0$ denotes directly computed results, and $k=4$ denotes rational approximations after four refinement steps. Electrode-potential errors are reported in table~\ref{tab:s22}.}\label{tab:s28}
\footnotesize\setlength{\tabcolsep}{3pt}
\begin{tabular*}{\textwidth}{@{\extracolsep{\fill}}llllll@{}}
\toprule
& \multicolumn{2}{c}{\makecell{Voltages relative\\to electrode 16}} & \multicolumn{2}{c}{\makecell{Adjacent-electrode\\voltage differences}} & \makecell{Common\\component}\\
\cmidrule(lr){2-3}\cmidrule(lr){4-5}\cmidrule(l){6-6}
Implementation & $k=0$ & $k=4$ & $k=0$ & $k=4$ & $k=0$\\\midrule
EIDORS augmented ($R$) & $1.404\times10^{-13}$ & $1.694\times10^{-13}$ & $2.747\times10^{-14}$ & $3.391\times10^{-14}$ & $1.104\times10^{-14}$\\
MATLAB Robin ($A$) & $1.708\times10^{-13}$ & $1.694\times10^{-13}$ & $3.427\times10^{-14}$ & $3.391\times10^{-14}$ & $2.586\times10^{-14}$\\
DOLFINx ($P_s$) & $2.401\times10^{-14}$ & $2.736\times10^{-14}$ & $7.491\times10^{-15}$ & $7.663\times10^{-15}$ & $1.730\times10^{-14}$\\
FreeFEM ($P_s$) & $7.884\times10^{-14}$ & $8.142\times10^{-14}$ & $1.649\times10^{-14}$ & $1.674\times10^{-14}$ & $1.314\times10^{-14}$\\
GetFEM ($P_s$) & $7.951\times10^{-14}$ & $1.210\times10^{-13}$ & $1.542\times10^{-14}$ & $2.437\times10^{-14}$ & $6.769\times10^{-15}$\\
MFEM ($P_s$) & $1.270\times10^{-13}$ & $4.203\times10^{-14}$ & $2.986\times10^{-14}$ & $8.016\times10^{-15}$ & $4.579\times10^{-14}$\\
NGSolve ($P_s$) & $7.282\times10^{-14}$ & $7.370\times10^{-14}$ & $1.882\times10^{-14}$ & $1.900\times10^{-14}$ & $1.518\times10^{-14}$\\
\bottomrule\end{tabular*}\end{table}

For the directly computed results, the errors in the voltages relative to electrode 16 range from $2.401\times10^{-14}$ to $1.708\times10^{-13}$, comparable to the electrode-potential errors in table~\ref{tab:s22}, and those in the adjacent-electrode voltages from $7.491\times10^{-15}$ to $3.427\times10^{-14}$. For every implementation, the common component is smaller than the maximum electrode-potential error, and it does not exceed $4.579\times10^{-14}$; the non-zero shift-invariant errors are therefore not explained by a common potential shift. After four refinement steps, the voltage-difference errors increase for EIDORS, FreeFEM and GetFEM and decrease for $A$ and MFEM; for NGSolve and DOLFINx they increase slightly, although their electrode-potential errors decrease. Voltage-difference errors, like potential errors, therefore need not decrease with the assembled-system residual.

\clearpage
\section{Reading the mathematical argument and Lean proofs together}
The main result is the equivalence of the ungrounded augmented CEM and the Robin--transadmittance procedure (main-text theorem 1). The formalisation covers the discrete energy argument, equivalence and reference results, including the zero-mean/pseudoinverse construction. Here lemma 1 supplies the energy structure; corollary 1 gives current compatibility and uniqueness up to a constant. Furthermore, lemma 2 recovers a deleted reference equation, proposition 1 gives the nodal implementation used for the EIDORS comparisons, and corollary 2 establishes reference-independent voltage differences. Supplementary proposition~\ref{prop:supp_mean_zero} gives the full zero-mean construction. The package uses Lean 4.24.0 and mathlib v4.24.0, with dependency revisions fixed in the supplied manifest.

\subsection{Finite input data, the weak form and energy}
The continuous Robin weak form is obtained by integration by parts in the main text. The formal proof starts at a finite representation of the P1 coefficients. Its inputs are the following: a nonempty finite node set and electrode set; a connected graph of nodal differences; an invertible transformation of those differences; a constant-preserving boundary interpolation map; a boundary-sample owner map covering every electrode; and strictly positive volume and contact weights. These are the fields of \path{CEM.DiscreteModel} in \path{CEM/Assembly.lean}. Neither $H$ invertibility nor positivity of $F_0$ is an input field.

The energy identity in main-text lemma 1 has the finite coefficient representation
\begin{equation}
\sum_q\kappa_q(\mathcal G \bm{u})_q^2+
\sum_b\tau_b\bigl((\mathcal T \bm{u})_b-U_{\ell(b)}\bigr)^2,
\qquad \kappa_q,\tau_b>0.
\label{eq:finite_energy}
\end{equation}
Here $\mathcal G$ forms the element gradient components, $\mathcal T$ interpolates boundary values, and $\ell(b)$ identifies the electrode containing sample $b$. The definitions below connect this weighted sum of squares to the weak form and P1 coefficients, and give the finite proof of lemma 1.

For difference index $q$, write
\[
(\Delta \bm{u})_q=u_{\mathrm{head}(q)}-u_{\mathrm{tail}(q)},\qquad
\mathcal G=R\Delta,\qquad R\text{ invertible}.
\]
Connectivity means that any two nodes are joined by the undirected equivalence closure of these edges, formalised by \path{DifferenceConnected}. The proof \path{incidence_kernel} follows such paths to show that $\Delta \bm{u}=0$ precisely when $\bm{u}$ is constant. Invertibility of $R$ gives the same statement for $\mathcal G \bm{u}=0$ through \path{DiscreteModel.bulk_kernel}. For P1 triangles, choose two vertex differences on each triangle and let the corresponding $2\times2$ block of $R$ convert them to gradient components. Nondegenerate triangle geometry makes each block invertible; no sign condition on stiffness off-diagonal entries is required. Each gradient component has weight $\int_T\sigma\,\mathrm dx>0$. Identifying these geometric inputs is an analytical argument, not a formal mesh-validation routine.

For boundary sample $b$, let $\ell(b)$ denote its electrode. Set
\[
(\mathcal S \bm{U})_b=U_{\ell(b)},\qquad
\mathcal T\bm1_N=\bm1_B,\qquad
W_\Omega=\operatorname{diag}(\bm{\kappa}),\quad
W_\partial=\operatorname{diag}(\bm{\tau}),
\]
where $\bm{\kappa}=(\kappa_q)_q$ and $\bm{\tau}=(\tau_b)_b$ collect the strictly positive volume and contact weights, and each electrode is represented by at least one sample. The formal block definitions are
\[
\begin{split}
K&=\mathcal G^{\mathsf T}W_\Omega\mathcal G,\qquad
A=\mathcal T^{\mathsf T}W_\partial\mathcal T,\qquad H=K+A,\\
C&=\mathcal T^{\mathsf T}W_\partial\mathcal S,\qquad
D=\mathcal S^{\mathsf T}W_\partial\mathcal S.
\end{split}
\]
In particular, the single-electrode ownership of each sample makes $D$ diagonal. Define
\[
\mathcal B=\begin{bmatrix}\mathcal G&0\\\mathcal T&-\mathcal S\end{bmatrix},\qquad
\mathcal W=\operatorname{diag}(\bm{\kappa},\bm{\tau}).
\]
The theorem \path{DiscreteModel.full_eq_augmented} expands $\mathcal B^{\mathsf T}\mathcal W\mathcal B$ into precisely $[H,-C;-C^{\mathsf T},D]$. The weighted bilinear identity in \path{CEM/Energy.lean} identifies testing this matrix against $(\bm{v},\bm{V})$ with
\[
\sum_q\kappa_q(\mathcal G\bm{v})_q(\mathcal G\bm{u})_q+
\sum_b\tau_b\bigl((\mathcal T\bm{v})_b-V_{\ell(b)}\bigr)
                   \bigl((\mathcal T\bm{u})_b-U_{\ell(b)}\bigr).
\]
Testing all coefficient pairs is equivalent to $F_0(\bm{u},\bm{U})=(0,\bm{I})$; \path{DiscreteModel.weak_form_iff} proves both directions, using coordinate test vectors for the forward implication. Taking the trial and test vectors equal gives \path{energy_identity}. Positive weights make zero energy equivalent to $\mathcal B(\bm{u},\bm{U})=0$. Connectivity then gives constant $\bm{u}$, and constant-preserving interpolation and electrode coverage give the same constant for every $U_\ell$. The results \path{energy_zero_iff}, \path{full_kernel_iff} and \path{H_posDef} implement these steps. The identities $K\bm1=0$, $A\bm1=C\bm1$ and $C^{\mathsf T}\bm1=D\bm1$ are checked as \path{K_constant}, \path{A_constant} and \path{coupling_constant}, all in namespace \texttt{CEM.DiscreteModel}.

\noindent\textbf{P1 boundary coefficients.}
For one straight electrode edge, take trace samples $(u_a,u_b,(u_a+u_b)/2)$ and weights $(t,t,4t)$, $t=h/(6z)>0$. Their contact energy is
\[
t(u_a-U)^2+t(u_b-U)^2+4t\left(\frac{u_a+u_b}{2}-U\right)^2
=2t\bigl[(u_a-U)^2+(u_a-U)(u_b-U)+(u_b-U)^2\bigr].
\]
Direct integration of the squared linear trace gives the same expression. The formal theorem \path{P1Boundary.local_matrix} checks that the weighted Gram matrix, in coordinates $(u_a,u_b,U)$, is
\[
t\begin{bmatrix}2&1&-3\\1&2&-3\\-3&-3&6\end{bmatrix}.
\]
Its blocks are exactly the main-text P1 edge formulas. The module also checks constant preservation, positivity of the three weights and the displayed energy identity. This checks the coefficient representation, rather than a particular platform's quadrature implementation.

\noindent\textbf{Transfer to the Robin matrix and to grounding.}
In \path{CEM/Schur.lean}, \path{DiscreteModel.assembled_iff_robin} derives $H$ invertibility from \path{H_posDef} before invoking the generic main theorem. The identity
\[
(Z_0\bm{U},\bm{U})^{\mathsf T}F_0(Z_0\bm{U},\bm{U})=\bm{U}^{\mathsf T}M\bm{U}
\]
is \path{robin_energy}. It transfers the energy kernel to $M$; \path{Z0_constant}, \path{M_constant}, \path{M_posSemidef}, \path{M_energy_zero_iff} and \path{M_kernel_iff} establish corollary 1's matrix properties. Its existence assertion follows constructively from supplementary proposition~\ref{prop:supp_mean_zero}; the full solution is \path{DiscreteModel.assembled_zeroMean_solution}. In \path{CEM/GroundedEnergy.lean}, inserting $u_g=0$ preserves the retained quadratic energy and removes the common constant. This proves $F_g,H_g,M_g$ positive definite. For balanced currents, \path{assembled_grounded_iff_full} recovers the deleted row, and \path{assembled_grounded_full_exists_unique} gives full-CEM existence and uniqueness directly from the finite model.

\subsection{Statements, notation and proof steps}
In \path{CEM/Equivalence.lean}, the finite index type \texttt{N} denotes all nodal degrees of freedom, \texttt{E} denotes electrodes, and \texttt{H,C,D} are the ungrounded blocks. The definitions \texttt{augmented H C D}, \texttt{response H C} and \texttt{transadmittance H C D} represent $F_0$, $Z_0=H^{-1}C$ and $M=D-C^{\mathsf T}Z_0$. The same definitions are later applied to $H_g,C_g,D$, with \texttt{N} then indexing the retained nodes; their values are $F_g,Z,M_g$. The change is an explicit instantiation of one algebraic result, not a different equivalence theorem. In the assembled-model namespace, \texttt{DiscreteModel.Z0} and \texttt{DiscreteModel.M} always mean the ungrounded $Z_0$ and $M$.

Lean writes a stacked vector $[\bm{u}^{\mathsf T},\bm{U}^{\mathsf T}]^{\mathsf T}$ as \texttt{Sum.elim u U}. Its two sides are indexed by \texttt{N} and \texttt{E} through their disjoint union. The hypothesis \texttt{Invertible H} supplies an actual two-sided inverse, rather than a floating-point inversion routine. The conclusion of \path{augmented_iff_robin} states precisely that the block equation holds if and only if $\bm{u}=Z_0\bm{U}$ and $M\bm{U}=\bm{I}$. Step 1 invokes \path{augmented_iff_rows}. Step 2 uses \path{response_equation}, which states $HZ_0=C$, and \path{nodal_iff_response}. Step 3 substitutes into the second row in both directions to complete \path{augmented_iff_robin}.

For reference changes, \texttt{Option J} adds one distinguished index to the retained index set \texttt{J}: \texttt{none} is the fixed reference and \texttt{some i} is a retained degree of freedom. This is an indexing convention and does not restrict where the reference lies in a mesh. \texttt{groundLift} inserts a zero potential, and \texttt{F.submatrix some some} deletes the corresponding row and column. The general theorem \path{reference_reduction_iff_full} keeps the full right-hand side $\bm{b}$, including its potentially non-zero entry at \texttt{none}. It therefore covers a fixed electrode potential with a prescribed non-zero current. The specialised \path{grounded_reduction_iff_full} inserts a zero right-hand-side entry at a nodal reference.

In the generic \path{CEM/Complete.lean}, \texttt{F} is the full node--electrode matrix, indexed by \texttt{Option (N} $\oplus$ \texttt{E)} with the reference at a node. The hypotheses \texttt{hsym}, \texttt{hconstant}, \texttt{hbalance} and \texttt{hblocks} mean, respectively,
\[
\begin{gathered}
F_0^{\mathsf T}=F_0,\qquad F_0\bm1=0,\qquad \sum_\ell I_\ell=0,\\
F_0\text{ with the ground row/column removed}=F_g.
\end{gathered}
\]
The full vector is \texttt{groundLift (Sum.elim v U)}, corresponding to $(P\bm{v},\bm{U})$. The hypothesis \texttt{hblocks} connects the supplied full matrix to the grounded blocks; it is not inferred from their names. The model-based \path{CEM/GroundedEnergy.lean} uses the equivalent ordering \texttt{(Option N)} $\oplus$ \texttt{E}: the ground is \texttt{Sum.inl none}, and its full vector is \texttt{Sum.elim (groundLift v) U}. This places the reference explicitly in the node set and derives the block and positivity facts from assembly. The two modules describe the same nodal deletion with different index orderings.

{\small
\def\UrlBreaks{\do\/\do\_\do\.}
\begin{longtable}{@{}>{\raggedright\arraybackslash}p{0.22\textwidth}>{\raggedright\arraybackslash}p{0.32\textwidth}>{\raggedright\arraybackslash}p{0.39\textwidth}@{}}
\caption{Detailed paper--Lean correspondence. All names are in namespace \texttt{CEM}; assumptions are mathematical hypotheses of the statements.}\label{tab:s29}\\
\toprule
Paper argument & Required hypotheses & Lean statements and source\\\midrule
\endfirsthead
\toprule Paper argument & Required hypotheses & Lean statements and source\\\midrule
\endhead
theorem 1, step 1: expand the block system & Same $H,C,D$ & \path{augmented_iff_rows}; \path{CEM/Equivalence.lean}\\\addlinespace[5pt]
theorem 1, step 2: identify the Robin solution & $H$ invertible & \path{response_equation}, \path{nodal_iff_response}; \path{CEM/Equivalence.lean}\\\addlinespace[5pt]
theorem 1, step 3: substitute in both directions & $H$ invertible; no inverse of $M$, symmetry or balance required & \path{augmented_iff_robin}; \path{CEM/Equivalence.lean}\\\addlinespace[5pt]
lemma 1 and section 2.2: finite weak form & Finite gradient and trace maps with positive weights & \path{DiscreteModel.variational_eq}, \path{DiscreteModel.weak_form_iff}; \path{CEM/Assembly.lean}\\\addlinespace[5pt]
lemma 1: block identities, energy and kernel & Connected differences, invertible gradient transform, constant-preserving trace, electrode coverage & \path{DiscreteModel.K_constant}, \path{A_constant}, \path{coupling_constant}, \path{energy_identity}, \path{full_kernel_iff}, \path{H_posDef}; \path{CEM/Assembly.lean}\\\addlinespace[5pt]
P1 edge representation for lemma 1 & Endpoint/midpoint interpolation, positive weights $t,t,4t$ & \path{P1Boundary.trace_partition}, \path{weights_positive}, \path{local_matrix}, \path{local_energy}; \path{CEM/P1Boundary.lean}\\\addlinespace[5pt]
corollary 1: complete conclusion (direct entry) & Finite-model inputs; balanced current for existence & \path{DiscreteModel.assembled_current_compatibility}; \path{CEM/MeanZero.lean}\\\addlinespace[5pt]
corollary 1: supporting properties of $Z_0$ and $M$ & Same finite-model inputs; $H$ invertibility derived from lemma 1 & \path{DiscreteModel.assembled_iff_robin}, \path{Z0_constant}, \path{M_constant}, \path{robin_energy}, \path{M_posSemidef}, \path{M_energy_zero_iff}, \path{M_kernel_iff}; \path{CEM/Schur.lean}\\\addlinespace[5pt]
corollary 1: constructed representative & Same finite-model inputs; balanced current & \path{DiscreteModel.assembled_zeroMean_solution}; \path{CEM/MeanZero.lean}; construction in supplementary proposition S1\\\addlinespace[5pt]
lemma 2: recover a deleted reference equation & Symmetric full matrix, constant null vector, balanced full load & \path{total_current_zero}, \path{recover_ground_row}, \path{groundLift_none}, \path{groundLift_some}, \path{reduced_mulVec}, \path{reference_reduction_iff_full}; \path{CEM/Grounding.lean}\\\addlinespace[5pt]
proposition 1: nodal deletion and Robin reconstruction & lemma 2's hypotheses; zero nodal load; correct grounded blocks; $H_g$ invertible & \path{grounded_reduction_iff_full}; \path{CEM/Grounding.lean}; \path{full_cem_iff_robin}; \path{CEM/Complete.lean}\\\addlinespace[5pt]
proposition 1: explicit unique solution & Additionally $M_g$ invertible & \path{augmented_iff_explicit_solution}, \path{grounded_exists_unique}; \path{CEM/Equivalence.lean}; \path{full_cem_iff_explicit_solution}, \path{full_cem_exists_unique}; \path{CEM/Complete.lean}\\\addlinespace[5pt]
proposition 1 from assembled data & Finite-model inputs, zero nodal reference and balanced current & \path{DiscreteModel.Fg_posDef}, \path{Hg_posDef}, \path{Mg_posDef}, \path{assembled_grounded_full_exists_unique}; \path{CEM/GroundedEnergy.lean}\\\addlinespace[5pt]
Section 2.4: auxiliary ground reaction & Symmetry, constant null vector and all retained equations & \path{ground_reaction_eq_neg_sum}; \path{CEM/Grounding.lean}\\\addlinespace[5pt]
corollary 2: same voltage differences for any full solutions & proposition 1's uniqueness conditions, same current input & \path{gauge_shift}, \path{gauge_voltage_difference}, \path{full_cem_solutions_equal_after_grounding}, \path{full_cem_solutions_same_voltage_differences}; \path{CEM/Gauge.lean}\\\addlinespace[5pt]
Supplementary proposition S1, steps 1--2: regularisation and pseudoinverse & Positivity and constant energy kernel of $M$, supplied by \path{DiscreteModel.conservativeEnergy} & \path{regularized_posDef}, \path{mul_cemInverse}, \path{cemInverse_mul}, \path{cemInverse_penrose}, \path{cemInverse_unique}; \path{CEM/MeanZero.lean}\\\addlinespace[5pt]
Supplementary proposition S1, step 3: zero-mean solution and unique shift & Balanced current; same finite-model inputs for the full solution & \path{zeroMean_solution_iff}, \path{DiscreteModel.assembled_zeroMean_formula}, \path{assembled_zeroMean_solution}, \path{CEM.zeroMean_unique}; \path{CEM/MeanZero.lean}\\
\bottomrule
\end{longtable}
}

Within table rows, names after \texttt{DiscreteModel.} or \texttt{P1Boundary.} retain that namespace unless a different source module is stated. The two elementary identities \path{groundLift_none} and \path{groundLift_some} additionally verify the zero reference value and preservation of the retained entries. The file \path{Audit.lean} checks the axiom dependencies of every project theorem, including the energy, pseudoinverse and exact-example modules. Only \texttt{propext}, \texttt{Classical.choice} and \texttt{Quot.sound}, or subsets of them, are accepted; there are no unfinished proofs or theorem-specific axioms.

The invariance proof deliberately has two stages. \path{gauge_shift} and \path{gauge_voltage_difference} show that a \emph{given} common shift preserves the equations and voltage differences. \path{full_cem_solutions_equal_after_grounding} then uses uniqueness under the common nodal reference to prove that \emph{any two} full solutions are related by such a shift. Thus reference independence is not deduced from the constant-nullspace identity alone.

\subsection{Constructing and identifying the pseudoinverse}
The main text uses a short formula for the zero-mean reference. Here we give its full construction, including the uniqueness of the pseudoinverse. This proof uses the matrix properties in main-text corollary 1, without assuming its existence conclusion.

\begin{supproposition}[Zero-mean solution and pseudoinverse]\label{prop:supp_mean_zero}
Let $M$ be the ungrounded electrode matrix of the exactly assembled CEM, with $L\geq1$, and set
\[
J=\frac1L\bm1_L\bm1_L^{\mathsf T},\qquad G=(M+J)^{-1}-J.
\]
Then $M+J$ is positive definite, $MG=GM=\mathrm{Id}_L-J$, and $G=M^\dagger$, the unique Moore--Penrose pseudoinverse. For every $\bm1_L^{\mathsf T}\bm{I}=0$, $\bm{U}=G\bm{I}$ is the unique solution of $M\bm{U}=\bm{I}$ with $\bm1_L^{\mathsf T}\bm{U}=0$, and $\bm{u}=Z_0G\bm{I}$ gives the corresponding unique full CEM solution with this reference.
\end{supproposition}
\begin{proof}
\emph{Step 1: remove the constant nullspace.} The identity
\[
\bm{x}^{\mathsf T}(M+J)\bm{x}=\bm{x}^{\mathsf T}M\bm{x}+L^{-1}(\bm1_L^{\mathsf T}\bm{x})^2
\]
is a sum of non-negative terms. If it vanishes, the first term makes $\bm{x}$ constant, and the second makes that constant zero since $L>0$. Thus $M+J$ is positive definite and its inverse exists.

\emph{Step 2: construct and identify the pseudoinverse.} Symmetry and $M\bm1_L=0$ give $MJ=JM=0$, while $J^2=J$. Hence $(M+J)^{-1}J=J(M+J)^{-1}=J$. Direct multiplication gives $MG=GM=\mathrm{Id}_L-J$ and $JG=GJ=0$. Consequently,
\[
MGM=M,\qquad GMG=G,\qquad (MG)^{\mathsf T}=MG,\qquad (GM)^{\mathsf T}=GM,
\]
the four Moore--Penrose equations. To establish uniqueness explicitly, suppose $X$ satisfies the same four equations. Since $JM=0$ and $MX$ is symmetric, $(MX)J=0$. It follows that
\[
MX=MX(\mathrm{Id}_L-J)=MXMG=MG=\mathrm{Id}_L-J.
\]
Similarly, symmetry of $XM$ and $MJ=0$ give $J(XM)=0$. Using $XMX=X$ yields $JX=0$. Therefore $(M+J)X=\mathrm{Id}_L-J$ and $X=(M+J)^{-1}(\mathrm{Id}_L-J)=G$.

\emph{Step 3: solve the balanced-current problem.} Since $J\bm{I}=0$, we have $M(G\bm{I})=\bm{I}$ and $J(G\bm{I})=0$. Conversely, if $M\bm{U}=\bm{I}$ and $J\bm{U}=0$, then $\bm{U}=(\mathrm{Id}_L-J)\bm{U}=GM\bm{U}=G\bm{I}$. Main-text theorem 1 gives $\bm{u}=Z_0\bm{U}$, including its uniqueness once $\bm{U}$ is fixed.
\end{proof}

\noindent\emph{Lean 4:} All following names are in \texttt{CEM}, in \path{CEM/MeanZero.lean}. The code represents $J$ by \path{meanProjector} and $G$ by \path{cemInverse}. The record \texttt{ConservativeEnergy M} collects positivity, $M\bm1=0$ and the constant energy kernel; \path{DiscreteModel.conservativeEnergy} proves each property from the finite assembly data.

Step 1 is \path{regularized_posDef}. Step 2 follows \path{mul_meanProjector}, \path{meanProjector_mul}, \path{meanProjector_idempotent}, the inverse--projector identities, \path{mul_cemInverse}, \path{cemInverse_mul}, \path{cemInverse_penrose} and \path{cemInverse_unique}. Step 3 is \path{cemInverse_solution} and \path{zeroMean_solution_iff}. The full statement \path{DiscreteModel.assembled_zeroMean_formula} says precisely
\[
\bigl[F_0(\bm{u},\bm{U})=(0,\bm{I})\ \text{and}\ \bm1^{\mathsf T}\bm{U}=0\bigr]
\quad\Longleftrightarrow\quad
\bm{u}=Z_0G\bm{I}\ \text{and}\ \bm{U}=G\bm{I}.
\]
The accompanying \path{assembled_zeroMean_solution} states existence and uniqueness. This constructive result supplies the representative invoked in main-text corollary 1; the latter's constant-nullspace statement then characterises all solutions.

Finally, write $\overline U=L^{-1}\sum_\ell U_\ell$. The results \path{zeroMean_sum}, \path{zeroMean_difference} and \path{zeroMean_unique} show that $\bm{U}-\overline U\bm1$ is the unique zero-mean common shift and preserves voltage differences. Applying the same shift to the nodal potentials preserves the full equations by \path{gauge_shift} in \path{CEM/Gauge.lean}.

\subsection{A small example with exact rational coefficients}
Consider a two-node, two-electrode algebraic model with
\[
K=\begin{bmatrix}1&-1\\-1&1\end{bmatrix},\quad A=C=D=I_2,
\quad H=\begin{bmatrix}2&-1\\-1&2\end{bmatrix}.
\]
The full coordinate order is $(u_1,u_2,U_1,U_2)$, and the augmented matrix is
\[
F_0=\begin{bmatrix}2&-1&-1&0\\-1&2&0&-1\\-1&0&1&0\\0&-1&0&1\end{bmatrix}.
\]
This example illustrates the algebra; it is not an additional finite element benchmark. The ungrounded Robin responses and electrode matrix are
\[
Z_0=\frac13\begin{bmatrix}2&1\\1&2\end{bmatrix},\qquad
M=\frac13\begin{bmatrix}1&-1\\-1&1\end{bmatrix}.
\]
For $\bm{U}=[1,0]^{\mathsf T}$, $\bm{u}=[2/3,1/3]^{\mathsf T}$ and $\bm{I}=[1/3,-1/3]^{\mathsf T}$; for $\bm{U}=[0,1]^{\mathsf T}$ the two entries are reversed. Both responses satisfy every nodal equation and have conserved currents. These substitutions are checked by \path{Example.ungrounded_first_response} and \path{Example.ungrounded_second_response} in \path{CEM/Example.lean}.

For the prescribed input $\bm{I}=[1,-1]^{\mathsf T}$, the following reference choices give
\[
\begin{array}{c|c|c|c}
\text{reference}&\bm{u}^{\mathsf T}&\bm{U}^{\mathsf T}&U_1-U_2\\\hline
\sum_\ell U_\ell=0&(1/2,-1/2)&(3/2,-3/2)&3\\
U_2=0&(2,1)&(3,0)&3\\
u_1=0&(0,-1)&(1,-2)&3
\end{array}
\]
The three substitutions into $F_0(\bm{u},\bm{U})=(0,\bm{I})$ are checked by \path{Example.mean_zero_solution}, \path{Example.electrode_reference_solution} and \path{Example.grounded_solution}. The potentials differ by common constants. In particular, choosing electrode 2 as the voltage reference does not set its prescribed current $I_2=-1$ to zero.

If the nodal reference is imposed before computing the auxiliary responses, then
\[
H_g=[2],\qquad C_g=\begin{bmatrix}0&1\end{bmatrix},\qquad
Z=\begin{bmatrix}0&1/2\end{bmatrix},\qquad
M_g=\begin{bmatrix}1&0\\0&1/2\end{bmatrix}.
\]
The first grounded response has $\bm{v}=[0]$, electrode currents $(1,0)$ and ground reaction $-1$; the second has $\bm{v}=[1/2]$, currents $(0,1/2)$ and reaction $-1/2$. These are checked by \path{Example.first_response} and \path{Example.second_response}. Solving $M_g\bm{U}=[1,-1]^{\mathsf T}$ gives $\bm{U}=(1,-2)^{\mathsf T}$ and $\bm{v}=[-1]$. The weighted reactions cancel: $-1+(-2)(-1/2)=0$. The final solution agrees with the ungrounded construction at the same nodal reference, although $M_g\ne M$ and the individual auxiliary problems differ.

\subsection{Running the proof package and interpreting its scope}
With the pinned Lean toolchain available, run \texttt{lake build} in \path{formal/}; the supplied Windows script \path{verify.ps1} uses Lake on the executable search path, or an explicit \texttt{-Lake} path. This builds the general results, the rational examples, the axiom checks and \path{Paper.lean}. The command \texttt{lake env lean Paper.lean} prints the formal statements in main-text result order, followed by supplementary proposition S1 and the exact example. The supplied \path{README.md} gives the dependency-cache command for a fresh installation and the same reading order.

The finite coefficient weak form, weighted energy, constant kernel, positive definiteness and zero-mean/pseudoinverse construction are included in the formal proof. The geometry-to-coefficient identification and continuous integration-by-parts, Sobolev and trace arguments remain outside Lean. The code does not certify the experimental assembly programs, the 2535-node reference coefficients, floating-point solvers or PDE convergence. Its real inverse definitions are marked \texttt{noncomputable}: the build checks exact mathematical proofs, rather than running a numerical EIT solver. The theorem about algebraic solution sets also applies to shared stored coefficients treated as exact numbers when the Robin matrix is invertible; reference changes additionally require their exact symmetry and constant-nullspace identities. This is why the experiments first align the reference and then distinguish solution and assembly effects.